\PassOptionsToPackage{table}{xcolor}
\documentclass[sigconf,nonacm]{acmart}

\setcopyright{none}
\usepackage{microtype}
\usepackage{graphicx}
\usepackage{subcaption}
\usepackage{booktabs}
\usepackage{enumitem}
\usepackage{pifont}
\usepackage{anyfontsize}
\usepackage{wasysym}
\usepackage{tablefootnote}

\newcommand{\full}{\CIRCLE}
\renewcommand{\partial}{\LEFTcircle}
\newcommand{\none}{\Circle}

\usepackage{amsmath}
\usepackage{mathtools}
\usepackage{amsthm}
\usepackage[capitalize,noabbrev]{cleveref}
\usepackage{listings}

\definecolor{jsonbg}{HTML}{F7F9FB}
\definecolor{jsonrule}{HTML}{D0D7DE}
\definecolor{jsonstr}{HTML}{1F5AA6}
\definecolor{jsonnum}{HTML}{116149}
\definecolor{jsonbool}{HTML}{7A2E83}
\definecolor{jsonkey}{HTML}{0E3C79}

\lstdefinelanguage{json-asiaccs}{
  morestring=[b]",
  showstringspaces=false,
  breaklines=true,
  breakatwhitespace=false,
  columns=fullflexible,
  literate=
   *{\\}{{\textbackslash}}{1}
    {/}{{/}}{1}
    {_}{{\_}}{1}
    {-}{{-}}{1}
    {:}{{\color{black}{:}}}{1}
    {,}{{,}}{1}
    {.}{{.}}{1}
    {true}{{{\color{jsonbool}true}}}{4}
    {false}{{{\color{jsonbool}false}}}{5}
    {null}{{{\color{jsonbool}null}}}{4},
}

\lstdefinestyle{asiaccs-json}{
  language=json-asiaccs,
  basicstyle=\ttfamily\footnotesize,
  backgroundcolor=\color{jsonbg},
  frame=single,
  rulecolor=\color{jsonrule},
  framesep=3pt,
  aboveskip=2pt, belowskip=2pt,
  numbers=left,
  numbersep=6pt,
  xleftmargin=8pt, xrightmargin=8pt,
  upquote=true,
  keepspaces=true,
  tabsize=2,
  prebreak=\mbox{\tiny$\hookleftarrow$},
  postbreak=\mbox{\tiny$\hookrightarrow$},
  stringstyle=\color{jsonstr},
}

\theoremstyle{plain}
\newtheorem{theorem}{Theorem}[section]

\newtheorem{corollary}[theorem]{Corollary}
\theoremstyle{definition}
\newtheorem{definition}[theorem]{Definition}
\newtheorem{assumption}[theorem]{Assumption}
\theoremstyle{remark}

\ccsdesc[500]{Security and privacy~Authentication}
\ccsdesc[300]{Security and privacy~Privacy-preserving protocols}
\ccsdesc[100]{Security and privacy~Systems security}
\ccsdesc[100]{Computing methodologies~Intelligent agents}

\keywords{autonomous agents, host-independent authentication, agent identity (AID), verifiable execution traces (VET), notarized TLS transcripts (Web Proofs), zero-knowledge proofs}

\title[VET Your Agent]{VET Your Agent: Towards Host-Independent Autonomy via Verifiable Execution Traces}

\author{Artem Grigor}
\affiliation{%
  \institution{University of Oxford}
  \city{Oxford}
  \country{United Kingdom}
}
\email{artem.grigor@cs.ox.ac.uk}

\author{Christian Schroeder de Witt}
\affiliation{%
  \institution{University of Oxford}
  \city{Oxford}
  \country{United Kingdom}
}
\email{christian.schroeder@eng.ox.ac.uk}

\author{Simon Birnbach}
\affiliation{%
  \institution{University of Oxford}
  \city{Oxford}
  \country{United Kingdom}
}
\email{simon.birnbach@cs.ox.ac.uk}

\author{Ivan Martinovic}
\affiliation{%
  \institution{University of Oxford}
  \city{Oxford}
  \country{United Kingdom}
}
\email{ivan.martinovic@cs.ox.ac.uk}

\renewcommand{\shortauthors}{Grigor et al.}

\begin{document}
\begin{abstract}
Recent advances in large language models (LLMs) have enabled a new generation of \emph{autonomous agents} that operate over sustained periods and manage sensitive resources on behalf of users. Trusted for their ability to act without direct oversight, such agents are increasingly considered in high-stakes domains including financial management, dispute resolution, and governance. Yet in practice, agents execute on infrastructure controlled by a host, who can tamper with models, inputs, or outputs, undermining any meaningful notion of autonomy.

We address this gap by introducing \emph{VET} (Verifiable Execution Traces), a formal framework that achieves \emph{host-independent authentication} of agent outputs and takes a step toward \emph{host-independent autonomy}. Central to VET is the \emph{Agent Identity Document (AID)}, which specifies an agent’s configuration together with the proof systems required for verification. VET is compositional: it supports multiple proof mechanisms, including trusted hardware, succinct cryptographic proofs, and notarized TLS transcripts (\emph{Web Proofs}).

We implement VET for an API-based LLM agent and evaluate our instantiation on realistic workloads. We find that for today’s black-box, secret-bearing API calls, Web Proofs appear to be the most practical choice, with overhead typically under $3\times$ compared to direct API calls, while for public API calls, a lower-overhead TEE Proxy is often sufficient. As a case study, we deploy a \emph{verifiable trading agent} that produces proofs for each decision and composes Web Proofs with a TEE Proxy. Our results demonstrate that practical, host-agnostic authentication is already possible with current technology, laying the foundation for future systems that achieve full host-independent autonomy.
\end{abstract}

\maketitle

\section{Introduction}
\label{sec:introduction}

Autonomous agents have a long history, from early robotics in the 1960s~\cite{nilsson1984shakey} and text-based systems such as ELIZA~\cite{weizenbaum1966ELIZAComputerPrograma}, to deep reinforcement learning agents mastering ATARI games in the 2010s~\cite{mnih2013PlayingAtariDeep}. Today, large language models (LLMs) have enabled a new wave of autonomous agents. These agents use natural language as an interface to perform tasks such as booking travel, resolving customer service tickets and managing digital assets~\cite{liu2023webarena, backlund2025VendingBenchBenchmarkLongTerm, AIOfficeCFO}. By integrating tools and retrieval, they extend into domains including finance, healthcare, governance, and software development. Advances in reasoning (e.g., chain-of-thought prompting~\cite{wei2022ChainofthoughtPromptingElicits}) and post-training have further increased their capability. Over fifty agentic systems were publicly released in 2024 alone, reflecting both their utility and accelerating adoption~\cite{wang2024SurveyLargeLanguage, casper2025AIAgentIndex}.

As these systems become more capable, a new class of \emph{personal agents} has emerged, where users delegate credentials to agents that act on their behalf across external services~\cite{south2025AuthenticatedDelegationAuthorized, IntroducingOperatorOpenAI}. Such delegation grants access to highly sensitive resources, effectively making the agent a co-inhabitant of the user’s digital life. While self-hosting is possible in principle, in practice most users rely on third-party providers for convenience and scalability. This creates a fundamental mismatch between expectation and reality: a \emph{personal} agent is no longer truly personal once its execution is mediated by an untrusted host. The host infrastructure, which controls the runtime, can silently tamper with outputs, substitute models, or impersonate the agent entirely, turning a trusted assistant into an adversarial insider and breaking the very trust assumptions that “personal” agents imply.

However, the issue extends beyond so-called personal agents. As reliance on delegation grows, agents are increasingly positioned as decision-makers for communities and shared assets. This has given rise to \emph{autonomous agents}, presented as neutral and objective actors in contexts such as arbitration~\cite{Chen2022RobotCourt,Horton2024ForcedRobotArbitration}, portfolio management (e.g., ai16z’s “AI-Agent-led VC”), or public-facing “agent CEOs” such as Truth Terminal~\cite{2024WhatTruthTerminal,WhatAi16zAI16Za}. Unlike personal agents, these systems cannot plausibly be self-hosted by individuals; by design they are operated as independent entities on third-party infrastructure. Yet the same fundamental weakness remains: each agent ultimately runs on a host that can override, impersonate, or selectively disclose its outputs. 

This gap between the \emph{perception of autonomy} and the \emph{reality of host control} has already led to real-world failures. Developers have impersonated their own agents for financial gain~\cite{joyce2024_ai_impersonation}; others allegedly manipulated outputs to advertise tokens or authorize salary payouts~\cite{miyahedge2024_ai_manipulation,andreessen2024_ai_grant}. These incidents demonstrate that when autonomy cannot be independently verified, hosts retain the power to subvert agents, and will act on that, when there are no checks in place, leaving trust in such systems inherently fragile.

\paragraph{Existing approaches.}
In response to this, both industry and academia have attempted to define how to authenticate agent outputs~\cite{chan2024VisibilityAIAgents, MCPIDocumentation, AuthGenAIAuth0, 2025Agent2AgentProtocolA2A}. Most proposals, however, rely on informal specifications or implicitly assume that the host is honest. Others take steps toward host-independence but anchor trust in a single mechanism such as TEEs~\cite{malhotra2024SettingYourPet,LaunchElizaV2,IntroducingWT3Trustless}. While these efforts are valuable, they remain narrow in scope: none provide a general framework that can accommodate diverse proof systems and still guarantee that an agent’s behavior is verifiable regardless of who controls the host. As a result, the central challenge of establishing \emph{host-independent autonomy} remains unresolved.

\begin{table}[t]
\centering
\small
\setlength{\tabcolsep}{6pt}
\renewcommand{\arraystretch}{1.1}
\begin{tabular}{p{3.5cm} c c p{4cm}}
\toprule
\textbf{Scheme} & \textbf{Host-Independent} & \textbf{Formalized} \\
\midrule
AI IDs (Chan et al.)~\cite{chan2024VisibilityAIAgents} & \none & \none  \\
MCP-Identity / A2A~\cite{MCPIDocumentation,2025Agent2AgentProtocolA2A} & \none & \partial  \\
Auth for GenAI~\cite{AuthGenAIAuth0} & \none & \none  \\
Delegation to Agents~\cite{south2025AuthenticatedDelegationAuthorized} & \none & \partial  \\
PET Rock~\cite{malhotra2024SettingYourPet} & \partial & \none   \\
\textbf{VET (this work)} & \full & \full  \\
\bottomrule
\end{tabular}
\caption{Comparison of agent identity/authentication schemes. Circles denote: \full~addressed, \partial~partially addressed, \none~not addressed.}
\label{tab:identity_schemes}
\end{table}

In this paper, we propose \emph{VET} (Verifiable Execution Traces), a framework for authenticating agents and taking a first step toward \emph{host-independent autonomy}. VET introduces a compositional authentication layer that (i) specifies an agent’s configuration and trust assumptions through an \emph{Agent Identity Document (AID)}, (ii) binds each output to a verifiable execution trace consistent with that AID, and (iii) abstracts deployment-specific details by modelling the agent as a set of verifiable components, each equipped with its own prover–verifier pair. This structure allows diverse verification techniques to be incorporated interchangeably within a single, unified identity, and shifts the trust anchor from \emph{who hosts the agent} to \emph{what configuration and proofs define the agent}.

We then focus on how to practically apply the VET framework to today’s agents. We observe that many deployed agent systems rely on \emph{third-party, black-box inference and tool APIs} while executing only relatively lightweight orchestration code on an agent's maintainer (host), as seen both in public-facing autonomous agents~\cite{2024WhatTruthTerminal,WhatAi16zAI16Za} and in competitive environments such as Alpha Arena~\cite{alpha_arena_2025}. In this deployment model, inference and tool APIs are consumed as black-box services and are therefore treated as trusted components of the system. The orchestration layer (prompting, tool wiring, and glue code) executed on the host then often constitutes the primary attack surface, as exemplified by recent industrial efforts on agent autonomy~\cite{malhotra2024SettingYourPet,phala2025_build_trustworthy_fintech_ai_agents}.

After considering multiple candidate component proof systems for this context, we identify notarized TLS transcript proofs (\emph{Web Proofs}) as a suitable choice for secret-bearing API components accessed over TLS (i.e., HTTPS), and TEE Proxies as a practical alternative for low-sensitivity or public-data APIs (e.g., price feeds). Both can be deployed without modifications to the API infrastructure and have overheads that scale with transcript size rather than model complexity. At the same time, we highlight that other proof systems are well suited to complementary settings, such as TEEs or ZKML for fully local or self-hosted agents where the computation execution itself must be authenticated.

Finally, we evaluate the resulting authentication overhead on realistic workloads of API-based LLM agent components, comparing Web Proofs against a no-proof baseline and a TEE Proxy deployment. Our results show that for secret-bearing calls where transcript privacy is beneficial, Web Proofs achieve authentication with modest overhead, typically below $3\times$ relative to direct API calls, while our optimized channel strategy amortizes handshake costs and supports multi-message sessions, with further optimizations in sight. Beyond benchmarks, we present \emph{VeriTrade}, a trading agent that makes periodic authenticated portfolio decisions, demonstrating how VET enables host-agnostic authentication in a realistic financial application and how different component proof systems can be composed in a single AID.

Together, these contributions establish both the formal basis and the practical feasibility of agent authentication. The remainder of the paper is organized as follows.

\paragraph{Paper organization.}
Section~\ref{sec:background} reviews identity, verifiable computation, and prior authentication approaches. Section~\ref{sec:threatmodel} defines our threat model and formalizes the distinction between authentication and autonomy. Section~\ref{sec:preliminaries} introduces necessary preliminaries. Section~\ref{sec:agent-authentication-framework} presents the VET framework and the Agent Identity Document abstraction. Section~\ref{sec:webproofs} details our Web Proofs instantiation, Section~\ref{sec:evaluation} reports benchmarks, and Section~\ref{sec:case-study} describes a real-world case study. Finally, Section~\ref{sec:limitations_and_future_work} discusses limitations and future work, and Section~\ref{sec:conclusion} concludes.

\section{Background}
\label{sec:background}

\subsection{Identity and Authentication for Software Programs}
\label{subsec:identity-for-software-programs}

Traditionally, software identity has been established through mechanisms such as code signing, package registries, notarized binaries, and more recently supply-chain attestations~\cite{anderson2020SecurityEngineering}. These methods bind artifacts to publishers before execution and attest to origin, but they do not guarantee that outputs reflect the intended program once execution begins.

That is what authentication of software programs seeks to establish. It ensures that the outputs of an execution correspond to the program’s declared identity. In practice, however, most approaches remain \emph{host-centric}: the platform is assumed to be in a trusted state and relied upon to report faithfully. Examples include secure boot and TPM-based remote attestation of measured boot or loaded binaries, as well as servers presenting endpoint certificates to authenticate TLS endpoints~\cite{lampson1992AuthenticationDistributedSystems, arbaugh1997SecureReliableBootstrap, sailer2004DesignImplementationTCGbased}. These mechanisms authenticate the host’s claims about code and configuration, but not the computation itself. If the host is malicious, it can still substitute programs, tamper with inputs, or forge outputs.

For autonomous agents this limitation is fundamental. It is not enough to authenticate \emph{who runs} a program; we must also authenticate \emph{what was executed} and bind each output to the declared configuration. This motivates approaches that shift authentication away from trusting the host and toward direct verification of computation, as we discuss next.

\subsection{Verifiable Computation for Software Authentication}
\label{subsec:verifiable-computation-for-software-execution}
Verifiable computation (VC) addresses whether a claimed output corresponds to the intended computation. The central idea is that a prover convinces a verifier that execution was correct without re-running the program. This concept originates from interactive proofs~\cite{goldwasser1985KnowledgeComplexityInteractive, babai1985TradingGroupTheory} and probabilistically checkable proofs (PCPs)~\cite{arora1998ProofVerificationHardness}, which established that any computation can be verified with limited checks.

Three main approaches to VC are widely studied.  
\emph{Succinct cryptographic proofs} such as Pinocchio, Groth16, PLONK, and STARKs~\cite{gennaro2013QuadraticSpanPrograms, groth2016SizePairingBasedNoninteractive, ben-sasson2018ScalableTransparentPostquantum, gabizon2019PLONKPermutationsLagrangebasesa} allow compact, efficiently verifiable proofs, optionally with zero-knowledge.  
\emph{Hardware-assisted attestation} uses trusted execution environments (TEEs) to prove that a binary executed on an untampered platform~\cite{costan2016IntelSGXExplained, schneider2022SoKHardwaresupportedTrusted}.  
\emph{Consensus-based and optimistic re-execution} replicate computation across committees or dispute-resolve on demand, providing integrity under majority or economic assumptions~\cite{bano2019SoKConsensusAge, conway2024OpMLOptimisticMachine}.

However, each approach falls short for autonomous agents: VC authenticates a \emph{single computation of a fixed program}, whereas agents are dynamic, multi-step systems that combine models, tools, and sometimes humans. This motivates a compositional framework that captures identity at the agent level, with VC techniques functioning as components of a complete solution that the framework provides.

\subsection{Agent Identity and Authentication}\label{subsec:agent-identity-and-authentication-protocols}
The question of how to identify an autonomous agent is newer than for conventional software and remains under-defined. Early provenance mechanisms proposed identifiers for model outputs~\cite{chan2024IDsAISystems}, linking generated content to the underlying model. These improved visibility but left security out of scope: impersonation and tampering were not addressed. The idea was later extended to visible \emph{AI IDs} for agents~\cite{chan2024VisibilityAIAgents}, intended to label agent-produced content in online ecosystems. However, these identifiers were voluntary and lacked both a formal threat model and security guarantees.

Industry efforts have also attempted to operationalize agent authentication. MCP-Identity~\cite{MCPIDocumentation}, Auth for GenAI~\cite{AuthGenAIAuth0}, and Google’s Agent-to-Agent (A2A) framework~\cite{2025Agent2AgentProtocolA2A} attach cryptographic credentials to service endpoints and describe capabilities via “agent cards.” In practice, these provide workflow primitives for authentication, deauthentication, and credential sharing, but all assume an honest host. Identities are thus fragile to host migration, and none of these approaches bind an agent’s declared configuration to the execution trace it produces.

Delegation highlights this gap further. South et al.~\cite{south2025AuthenticatedDelegationAuthorized} proposed delegated credentials for personal agents acting across services. Yet neither this nor prior work ensures that actions under delegation were genuinely produced by the agent rather than forged by the host.

Several experimental systems attempted to strengthen identity with verifiable computation. The PET Rock prototype~\cite{malhotra2024SettingYourPet} demonstrated enclave-backed agents but remained limited in scope. Other initiatives (e.g., ElizaOS on Phala TEEs~\cite{FeatAddedNew}) focused on feasibility rather than formal identity definitions.

\subsection{Agent Autonomy}
\label{subsec:agent-autonomy}

Unlike software identity and verifiable computation, the notion of \emph{autonomy} lacks a precise formalization. In industry usage, an agent is “autonomous” if it can reason and act without direct human oversight, and if no external party is assumed to significantly influence its decision process. At minimum, meaningful autonomy requires independence from the host: if the host operator can override, impersonate, or manipulate outputs, then autonomy is nominal rather than real. To make autonomy substantive in settings where assets are at stake, one must be able to \emph{authenticate} that observed actions indeed originate from the declared agent and not its host.

\begin{table}[t]
\centering
\footnotesize
\setlength{\tabcolsep}{6pt}
\renewcommand{\arraystretch}{1.1}
\begin{tabular}{l c p{3.5cm}} % 3 columns only
\toprule
\textbf{Platform} & \textbf{Host-Independent} & \textbf{Details} \\
\midrule
ElizaOS~\cite{2025ElizaOSEliza} & \partial & Phala TEE; zkTLS adapter \\
GAME\tablefootnote{\url{https://app.virtuals.io}} & \none  & Operator controls execution \\
0G\tablefootnote{\url{https://0g.ai}} & \partial & Consensus inference; early stage \\
FetchAI\tablefootnote{\url{https://fetch.ai}} & \none & Operator controls execution \\
SingularityNET\tablefootnote{\url{https://github.com/singnet}} & \none & Operator controls execution \\
Olas\tablefootnote{\url{https://olas.network}} & \partial & Operators with optional consensus \\
Cainam Ventures\tablefootnote{\url{https://www.cainamventures.com}} & \none & Operator controls execution \\
\bottomrule
\end{tabular}
\caption{“Autonomous agent” platforms. All emphasize deployment and monetization, but none provide host-independent or formally defined autonomy. Circles denote: \full~addressed, \partial~partially, \none~not.}
\label{tab:agent_kits}
\end{table}

Despite this, many frameworks market themselves as supporting “autonomous agents.” They emphasize deployment and monetization while leaving host trust largely unresolved. ElizaOS~\cite{2025ElizaOSEliza}, widely used in Web3, exemplifies this: it remains the default toolkit despite a trust model that requires to trust the host, and only some experiments using Phala TEEs integration and experimental zkTLS adapter~\cite{FeatAddedNew, LaunchElizaV2}. \emph{GAME} supports token-based co-ownership but leaves execution under operator control. \emph{0G} introduces consensus-backed inference, but adoption is nascent and costs remain high~\cite{fernandez-becerra2024EnhancingTrustAutonomous}. Marketplaces such as \emph{FetchAI}, \emph{SingularityNET}, and \emph{AgentVerse} let developers publish agents, but users must ultimately trust the operator. \emph{Olas} introduces “operators” with optional consensus gadgets, though practical deployments are rare. Newer players such as \emph{Cainam Ventures} emphasize monetization while leaving infrastructure trust explicitly out of scope.

\paragraph{Summary.}
In summary, today’s so-called \emph{autonomous agents} are autonomous only in name: their execution remains subject to host control, and their autonomy cannot be independently verified. Existing approaches to software identity and verifiable computation provide important building blocks but do not offer a general, plug-and-play solution for agents. This motivates the focus of this work on \emph{host-independent authentication}, a property where an agent’s outputs can be authenticated against its declared configuration and execution trace, regardless of the host. In this paper we take a first step by proposing a framework that achieves host-independent authentication of agent outputs, moving closer toward full host-independent autonomy.

\section{Threat Model}
\label{sec:threatmodel}

\subsection{System Model}
\label{subsec:system-model}

\begin{figure}[h]
  \centering
  \includegraphics[width=0.8\linewidth]{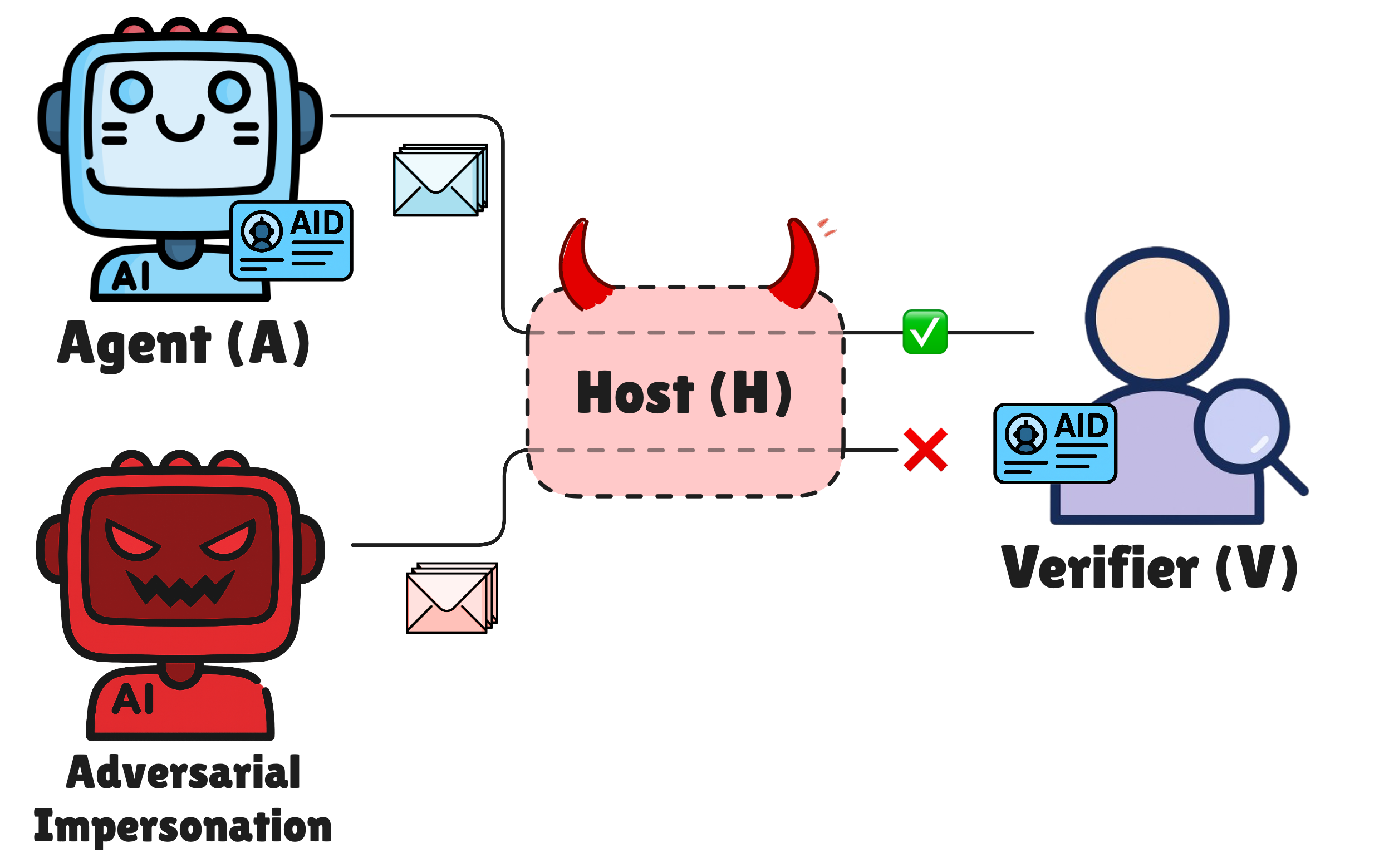}
  \vspace{-0.8em}
  \caption{System model. The verifier, given an AID, checks that an output is consistent with the declared agent configuration, even if the host controlling execution is malicious.}
  \Description{Diagram showing a verifier with an AID checking outputs against an agent configuration, despite a potentially malicious host.}
  \label{fig:threat-model}
\end{figure}

\begin{enumerate}[leftmargin=1.2em, label=\textbf{(\arabic*)}]
    \item \textbf{Autonomous Agent $(A)$:}
    A software entity uniquely described by an Agent Identity Document (AID). $A$ processes inputs and produces outputs consistent with its declared configuration.

    \item \textbf{Host $(H)$:}
    Controls the runtime environment of $A$. We treat $H$ as the primary adversary: fully malicious but computationally bounded. $H$ may replace models, tamper with inputs, fabricate or suppress outputs, or selectively disclose results (“cherry-picking”). It cannot break standard cryptography or compromise external components assumed correct (see Assumption~\ref{assumption:component-integrity}).

    \item \textbf{Verifier $(V)$:}
    Executes the verification protocol given an output and agent's AID. $V$ is assumed to follow the protocol faithfully but may be \emph{honest-but-curious}, attempting to extract sensitive information used by host to run the agent (e.g., API keys, proprietary prompts). Privacy-preserving design is therefore a primary requirement.
\end{enumerate}

\subsection{Security Goals}
\label{subsec:security-goals}

The framework addresses the problem of authenticating agent outputs in the presence of a malicious host. We distinguish two related objectives:

\paragraph{(1) Authentication of outputs (achieved).}
Given an AID, a verifier can check that an observed output corresponds to an execution of agent $A$ consistent with its declared configuration. We require:
\begin{itemize}[leftmargin=1.2em]
    \item \textbf{Soundness:} A malicious host $H$ cannot convince $V$ to accept an output–proof pair unless the output is consistent with the configuration in the AID.
    \item \textbf{Completeness:} For an honest execution of $A$, $V$ will accept the output–proof pair as valid.
    \item \textbf{Privacy:} Proofs reveal only what is strictly necessary. Leakage of secrets (e.g., API keys, prompts) is mitigated through selective disclosure or full zero-knowledge property.
\end{itemize}

\paragraph{(2) Host-independent autonomy (aspirational).}
Authentication alone ensures that individual outputs are genuine but does not protect against the host-level bias. We define \emph{host-independent autonomy} as the property that, for every execution trace, all outputs observable by a verifier can be authenticated against the AID, and $H$ cannot significantly influence $A$’s decisions by covert means (timing manipulation, suppression, or side-channel attacks). Our framework does not yet achieve this stronger property. Instead, we view authentication as a necessary first step. Achieving full host-independent autonomy remains future work.

\subsection{Security Assumptions}
\label{subsec:security-assumptions}

\begin{assumption}[Secure Primitives]
Standard primitives (e.g., digital signatures, hash functions) are secure against both $H$ and $V$.
\end{assumption}

\begin{assumption}[Component Integrity]
\label{assumption:component-integrity}
Each external component listed in the AID (e.g., APIs, oracles, LLM services) behaves according to its specified interface. Handling faulty components or collusion between $H$ or $V$ and such components is out of scope.
\end{assumption}

\begin{assumption}[Denial-of-Service]
We do not guarantee liveness. $H$ may suppress or delay outputs and proofs. Our guarantees apply only to results that are eventually delivered.
\end{assumption}

\begin{assumption}[Host Timing Advantage]
$H$ may observe, delay, or reorder outputs before forwarding them to $V$.
\end{assumption}

\section{Preliminaries}
\label{sec:preliminaries}

To reason about authentication, we require a minimal but formal definition of an LLM-based agent. While industry descriptions (e.g., OpenAI’s “LLM configured with instructions and tools”\footnote{\url{https://openai.github.io/openai-agents-python/agents/}}; Anthropic’s tool-using agents~\cite{anthropic2025IntroducingModelContext}) emphasize capabilities, they lack the formality needed for security analysis. We adopt the following abstraction, consistent with classic Sense–Think–Act models~\cite{wooldridge1995intelligent}.

Let $\Sigma^\ast$ be a shared finite string-based domain (e.g., UTF-8).

\begin{definition}[Tool]
    \label{def:tool}
A \emph{tool} is a named function
\[
   \mathsf{tool}_t : \Sigma^\ast \to \Sigma^\ast
\]
indexed by an identifier $t \in \mathcal{V}_{\mathrm{tool}}$, where $\mathcal{V}_{\mathrm{tool}}$ is a fixed vocabulary enumerating all recognized tools. Tools represent external APIs, local modules, or other services the agent may invoke.
\end{definition}

\begin{definition}[Core]
    \label{def:core}
The \emph{core} of an agent is a (probabilistic) function
\[
   \mathsf{Core} : \Sigma^\ast \to \Sigma^\ast \times \mathcal{P}(\mathcal{V}_{\mathrm{tool}}\times\Sigma^\ast),
\]
which, given a transcript $h$, produces (i) an output string $y$ and (ii) a set of tool invocations $(t,x)$.
\end{definition}

\begin{definition}[LLM-based Agent]
An \emph{agent} is a pair $\mathcal{A}=(\mathsf{Core}, \mathcal{T})$, where $\mathsf{Core}$ is a reasoning LLM and $\mathcal{T}$ is a set of authorized tools. The agent executes iteratively: at each step, the core processes the transcript, emits $y$, and invokes tools from $\mathcal{T}$, whose results are appended to the transcript.
\end{definition}

\begin{definition}[Execution Trace]
    \label{def:trace}
An \emph{execution trace} $\tau$ of agent $\mathcal{A}$ is the finite sequence of steps
\[
   \tau = \bigl( (y^{(0)}, T^{(0)}), (y^{(1)}, T^{(1)}), \dots, (y^{(n)}, T^{(n)}) \bigr),
\]
where each $T^{(j)} = \{(t_i, x_i, r_i)\}$ records the tool calls and results made at step $j$. Authentication later binds outputs $y^{(j)}$ (and inputs $x_i$) to a valid trace under $\mathcal{A}$.
\end{definition}

\begin{figure}[ht]
    \centering
    \includegraphics[width=0.82\linewidth]{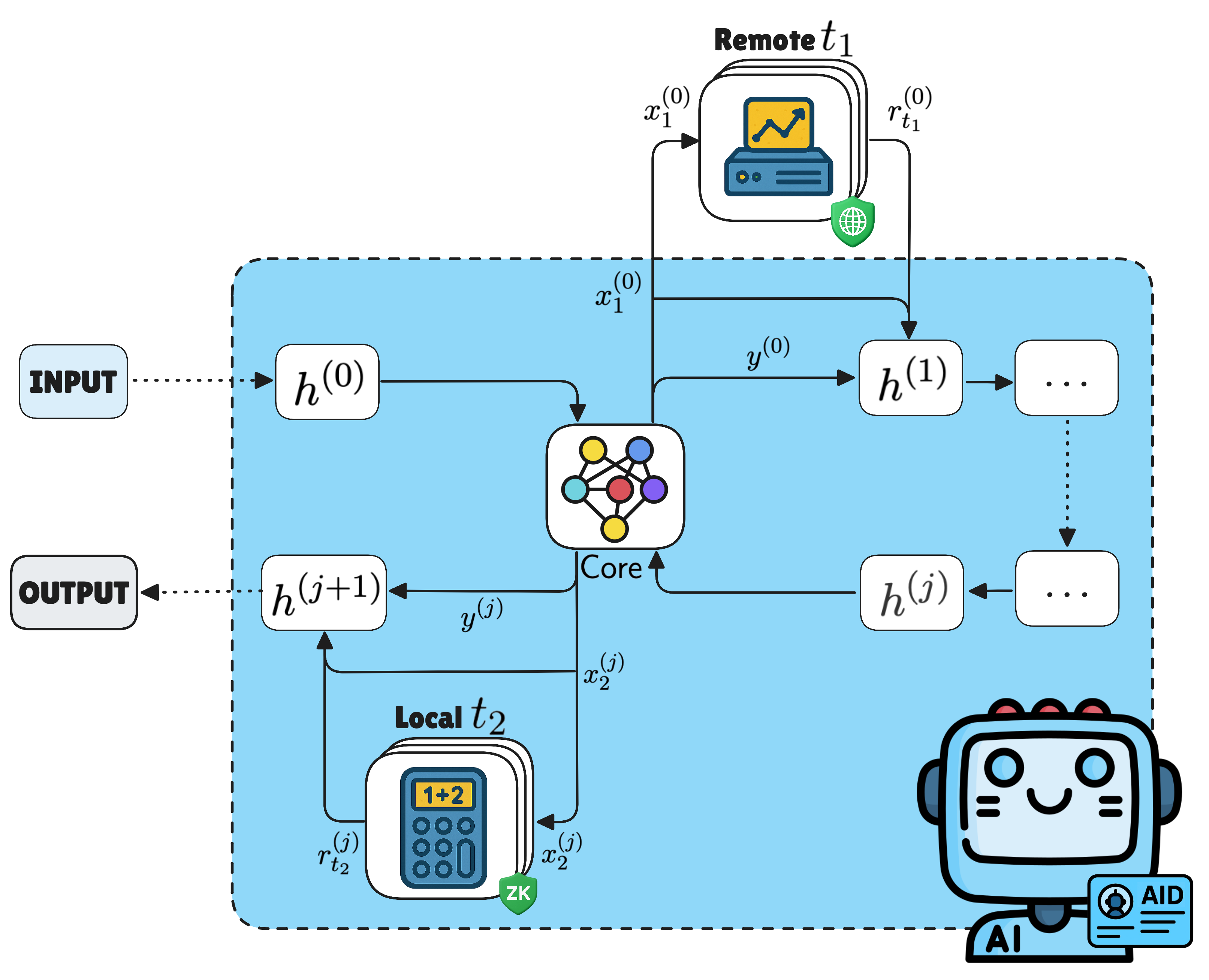}
    \caption{Illustration of the agent's execution loop. The agent's cognitive $\mathsf{Core}$ iteratively processes execution transcripts $h^{(j)}$, produces plaintext outputs $y^{(j)}$, invokes specified tools ($t_1$, $t_2$), and incorporates their responses ($r_{t_i}$) into subsequent execution transcripts.}
    \Description{Flow diagram of an agent execution loop: the core processes transcripts, outputs plaintext, invokes tools, and feeds tool responses back into the transcript for the next iteration.}
    \label{fig:agent-execution}
\end{figure}
\section{Agent Authentication Framework}
\label{sec:agent-authentication-framework}

\newcommand{\Tool}{\textsf{Tool}}
\newcommand{\Component}{\textsf{Component}} % alias for retro-compat.
\newcommand{\Core}{\textsf{Core}}
\newcommand{\negl}{\mathsf{negl}}
\newcommand{\Adv}{\textsf{Adv}}
\newcommand{\A}{\textsf{A}}
\newcommand{\Prove}{\mathsf{Prove_{(\mathcal{A},\Adv_\A)}}}
\newcommand{\Verify}{\mathsf{Verify_{(\mathcal{A},\Adv_\A)}}}
\newcommand{\OutGeneral}{\mathrm{Out}(\mathinner{\ldotp\!\ldotp\!\ldotp},\tau)}
\newcommand{\AID}{\textsf{AID}}
\newcommand{\Hash}{\textsf{Hash}}
\newcommand{\serialize}{\textsf{serialize}}
\subsection{Formal Model of Agent Authentication}
\label{subsec:formal-model}

Having formalized an agent $\mathcal{A}=(\mathsf{Core},\mathcal{T})$ and its execution traces (Definition~\ref{def:trace}), we now define what it means to \emph{authenticate} an output. Intuitively, a verifier should accept a candidate output $m \in \Sigma^\ast$ if and only if there exists a valid execution trace of $\mathcal{A}$ in which $m$ appears either as a plaintext output or as a tool input that must be authenticated by the corresponding tool provider.\footnote{We require that $m$ be generated by an actual run of $\mathcal{A}$, not merely by constructing a plausible-looking trace offline. See~\cite{wolf2024FundamentalLimitationsAlignment} for discussion of steering LLMs toward arbitrary outputs.}

\paragraph{Language of authentic outputs.}
Let $\mathrm{ValidTrace}_{\mathcal{A}}(\tau)$ be a polynomial-time predicate that checks whether $\tau$ is a valid execution trace of $\mathcal{A}$. Define two selectors:
\[
\mathsf{Out}_p(j,\tau) = y^{(j)} \quad\text{and}\quad
\mathsf{Out}_t(j,i,\tau) = x_i^{(j)}.
\]
The language of authentic outputs for $\mathcal{A}$ is:
\[
\mathcal{L}_{\mathcal{A}} =
\bigl\{ m \in \Sigma^\ast \,\big|\,
  \exists \tau :
  \begin{aligned}[t]
    & \mathrm{ValidTrace}_{\mathcal{A}}(\tau) = 1 \\
    & \wedge\; m \in \{\mathsf{Out}_p(j,\tau), \mathsf{Out}_t(j,i,\tau)\}
  \end{aligned}
\bigr\}.
\]

\begin{definition}[Agent-Authentication Scheme]
\label{def:auth-scheme}
For a fixed agent $\mathcal{A}$ and language $\mathcal{L}_{\mathcal{A}}$, an authentication scheme consists of two PPT algorithms:
\[
\begin{aligned}
 \Prove(1^\lambda,m,\tau) & \to \pi
\\
 \Verify(1^\lambda,m,\pi) & \to \{0,1\},
\end{aligned}
\]
where $\pi$ is a proof that $m \in \mathcal{L}_{\mathcal{A}}$ relative to trace $\tau$, and $\lambda$ is the security parameter, which describes the adversarial running time and negligible bounds.
\end{definition}

\begin{definition}[Completeness]
\label{def:completeness}
The scheme is \emph{complete} if for every valid trace $\tau$ and every $m \in \mathcal{L}_{\mathcal{A}}$, we have
\[
\Pr\bigl[\Verify(1^\lambda,m,\Prove(1^\lambda,m,\tau))=1\bigr] = 1.
\]
\end{definition}

\begin{definition}[Soundness]
\label{def:soundness}
The scheme is \emph{sound} if no PPT adversary $H^\ast$ can cause the verifier to accept an inauthentic output with non-negligible probability. Formally:
\[
\begin{aligned}
\Pr\bigl[
   (m,\pi)\leftarrow H^\ast(1^\lambda) :
   m \notin \mathcal{L}_{\mathcal{A}}
   \wedge \Verify(1^\lambda,m,&\pi)=1
\bigr]
 \\
& \le \negl(\lambda)\footnotemark .
\end{aligned}
\]
\end{definition}
\footnotetext{$\negl(\lambda)$ denotes a negligible function: for every polynomial $p(\lambda)$ there exists $\lambda_0$ such that for all $\lambda > \lambda_0$, $\negl(\lambda) < \tfrac{1}{p(\lambda)}$.}

\paragraph{Privacy.}
Execution traces often contain sensitive information irrelevant to authentication. We therefore require that proofs reveal only what is strictly necessary for verifying an output.

\begin{definition}[Minimal Disclosure]
\label{def:minimal-disclosure}
An authentication scheme satisfies \emph{minimal disclosure of information $\sigma$} if, for a specific sensitive value $\sigma$ and any valid proof $\pi$ of $m$, the verifier cannot learn $\sigma$ from $\pi$ beyond what is already derivable from $(AID,m)$. In other words, the scheme allows one to prove that an output $m$ is valid while hiding designated parts of the underlying execution trace.
\end{definition}

The strongest notion is zero-knowledge, where proofs hide \emph{all} information about the execution trace beyond the validity of $m$.

\begin{definition}[Zero-Knowledge]
\label{def:zk}
An authentication scheme is \emph{zero-knowledge} if the verifier learns nothing from a proof beyond the fact that $m$ is a valid output of $\mathcal{A}$. Formally, this requires the existence of a PPT simulator $\mathcal{S}$ such that for every valid $m \in \mathcal{L}_{\mathcal{A}}$, the distributions
\[
\pi \gets \Prove(1^\lambda,m,\tau)
\quad\text{and}\quad
\pi' \gets \mathcal{S}(1^\lambda,m)
\]
are computationally indistinguishable to any PPT adversary.
\end{definition}

\paragraph{Summary.} Together, Definitions~\ref{def:completeness}–\ref{def:zk} formalize our security goals: completeness (honest runs succeed), soundness (forgery is infeasible), and privacy (proofs either hide sensitive information or further reveal nothing beyond authenticity).
\subsection{Compositional Agent Authentication}
\label{subsec:comp-auth}

Having defined the authentication framework, we now turn to its instantiation. Our approach is to compose proofs of correctness for each agent component, the cognitive \Core{} and its authorized \Tool{}s, into an overall authentication proof. This bottom-up construction reduces authentication of the entire agent to authentication of its parts, and naturally aligns with the modular definition of agents introduced earlier.

\paragraph{Component proof systems.}
For each tool identifier $t \in \mathcal{V}_{\mathrm{tool}}$, define the relation
\[
R_t = \{(x,r) \mid r = \mathsf{tool}_t(x)\}.
\]
For the cognitive core:
\[
R_{\Core} = \{(h,(y,T)) \mid (y,T) = \Core(h)\}.
\]
A \emph{component proof system} for relation $R$ is a pair of PPT algorithms $(P,V)$ with perfect completeness and $\negl(\lambda)$-soundness against adversaries $\Adv_R$.

\paragraph{Compositional construction.}
Given a valid trace $\tau$, the prover $\Prove_\A$ generates component proofs $\pi_{\Core},\pi_t$ for each step and packages them into a composite proof $\pi$. The verifier $\Verify$ checks (i) that every component sub-proof is valid and (ii) that the reconstructed transcript is consistent. If both hold and the claimed message $m$ appears as an output in the trace, $\Verify$ accepts.

\begin{theorem}[Composition]
\label{thm:composition}
If each component proof system $(P,V)$ is complete and sound against its adversary class, then the compositional scheme $(\Prove_\A,\Verify)$ is complete and sound against
\[
\Adv_\A = \Adv_{\Core} \cap \bigcap_{t\in\mathcal{T}} \Adv_t.
\]
\end{theorem}

\begin{proof}[Proof sketch] % TODO - add a full proof in the appendix
Completeness follows directly: if all component proofs are valid under honest execution, then $\Verify$ accepts. Soundness is by contradiction: if $\Verify$ accepts an inauthentic output, given the transcript consistency, some sub-proof must have been forged, contradicting the soundness of the corresponding component. Since the transcript has bounded length, applying a union bound shows that the overall forgery probability remains negligible when the individual error bounds are combined.

\end{proof}

\paragraph{Privacy.} % TODO - add a full description to the appendix
In our compositional construction, privacy properties are inherited from the underlying component proof systems. If a component proof supports minimal disclosure or zero-knowledge, then the overall scheme achieves the same for that component’s contribution to the execution trace. In particular, minimal disclosure can be obtained by selectively hiding sensitive steps (e.g., API keys or prompts) whenever the component proof allows it. Stronger guarantees are possible by employing recursive proof techniques, which compress multiple execution steps into a single higher-level proof. This enables entire segments of a trace to be hidden while still proving global consistency, often with lower marginal cost than proving each step directly (e.g., folding SNARKs over another verifiable proofs). While we do not develop these recursive constructions further here, they provide a path toward achieving full zero-knowledge guarantees over complete execution traces when desired.
\subsection{Agent Identity Document (AID)}
\label{subsec:agent-identity}

The composition theorem (Thm.~\ref{thm:composition}) shows how to build an authentication scheme for any agent $\A=(\Core,\mathcal{T})$. To make this usable in practice, however, we require a stable, verifiable object that acts as a reference point for authentication. We call this the \emph{Agent Identity Document (AID)}. The AID plays the role of a “passport” for agents: it not only describes their configuration, but also binds that description to verifiable metadata so that third parties can independently check outputs against the declared identity. Without such a document, proofs of individual components could not be tied back to a unique agent, and authentication would collapse into ad hoc verification.

\begin{definition}[Agent Identity Document]
\label{def:agent-vid}
An \emph{AID} for an agent $\A=(\Core,\mathcal{T})$ is a tuple
\[
\AID = \bigl((\hat{\Core}, \hat{V}_{\Core}),\, \{(\hat{t}, \hat{V}_t)\}_{t \in \mathcal{T}}\bigr),
\]
where each component is specified by descriptive metadata ($\hat{\Core},\hat{t}$) and corresponding verification metadata ($\hat{V}_{\Core},\hat{V}_t$) that deterministically instantiate its proof system verifier.
The agent’s \emph{ID} is the collision-resistant hash of the serialized AID:
\[
  \mathsf{ID}_{\A} \coloneqq \Hash(\serialize(\AID)).
\]
\end{definition}

\paragraph{Relation to Agent Cards.}
The AID generalizes prior “Agent Cards”~\cite{2025Agent2AgentProtocolA2A,chan2024VisibilityAIAgents}, which previously only described agents capabilities. By explicitly including verification metadata, the AID enables not only description but also host-independent authentication of such agents.

\begin{corollary}[AID sufficiency]
Given an AID, one can reconstruct the verifier $V_{\!\A}$ for the compositional authentication scheme of agent $\A$. Thus the AID acts as both an identity and a certificate of verifiability.
\end{corollary}

\paragraph{Example.}
For instance, a trading agent’s AID may reference a GPT-4o core verified via TLS-Notary, alongside a price-feed API verified by consensus. Each entry specifies identifiers (model, API endpoint) and verification metadata (e.g., notary public key, consensus parameters).

\begin{figure}[t]
  \centering
    \begin{minipage}{0.98\linewidth}
      \lstset{style=asiaccs-json}
\begin{lstlisting}
{
  "agent_name": "VeriTradeBot",
  "core": {
    "model": "gpt-4o-2024-05-13",
    "endpoint": "https://api.openai.com/completions",
    "injection_algorithm_uid": "sha256:8afc9b24...",
    "parsing_algorithm_uid": "sha256:e34f8d97...",
    "verification": {
      "TLSNotary": {
        "protocol_version": "v0.1.0-alpha.10",
        "notary_public_key": "ecdsa-p256:04a1b2c3..."
      }
    }
  },
  "tools": [
    {
      "name": "PriceFeedAPI",
      "endpoint": "https://api.coingecko.com/api/v3/simple/price",
      "injection_algorithm_uid": "sha256:bad42015...",
      "parsing_algorithm_uid": "sha256:b932ffea...",     "verification": {
        "Consensus": {
          "protocol": "HotStuff-v1.1",
          "committee_id": "pool-eth-prices-mainnet-001",
          "committee_size": 7,
          "fault_tolerance": "f = 2"
        }
      }
    }
  ],
  "agent_hash": "sha256:27c8f3d8b9a4..."
}
\end{lstlisting}
    \end{minipage}
  \caption{Sample Agent Identity Document (AID) for \emph{VeriTradeBot}.}
  \Description{Code listing showing the Agent Identity Document.}
  \label{fig:sample-aid}
\end{figure}

\section{Candidate Proof Systems}
\label{sec:proof-systems-limitations}

Our framework is agnostic to how each component produces its local proof, provided that the prover–verifier pair declared in the Agent Identity Document (AID) satisfies the completeness and soundness properties defined in Section~\ref{subsec:security-goals}. In practice, three main families of proof systems appear suitable as component proof systems: (i) succinct cryptographic proofs, (ii) trusted execution environments (TEEs), and (iii) consensus-based re-execution. Each provides useful properties in specific contexts, but as we show, none is sufficient for today’s API-based LLM agents.

\paragraph{Succinct cryptographic proofs.}
Succinct non-interactive arguments of knowledge (SNARKs/STARKs)~\cite{groth2016SizePairingBasedNoninteractive,gennaro2013QuadraticSpanPrograms,ben-sasson2018ScalableTransparentPostquantum} enable compact, efficiently verifiable proofs of computation, with strong soundness and optional zero-knowledge. However, proof generation for large-scale inference remains prohibitively expensive. Even optimized ZKML systems~\cite{peng2025SurveyZeroKnowledgeProof,chen2024ZKMLOptimizingSystem} only scale to small- or medium-sized neural networks (e.g., GPT-2–class models). In practice, succinct proofs are therefore most suitable for local tools and smaller sub-models within an agent, rather than as the primary mechanism for verifying billion-parameter LLM cores.

\paragraph{Trusted execution environments (TEEs).}
Hardware enclaves such as Intel SGX and TDX~\cite{schneider2022SoKHardwaresupportedTrusted,costan2016IntelSGXExplained} provide attestation that a specific binary is executed in an untampered environment, while also preserving the confidentiality of the executed code and data from the TEE owner. Modern TEE hardware can support increasingly large workloads, including state-of-the-art LLM inference using the latest Confidential GPUs~\cite{2025NVIDIAConfidentialComputing,el-hindi2022BenchmarkingSecondGeneration,PerformanceConsiderationsHardwareIsolated}. TEEs are thus well suited when agents run \emph{local} inference over open-source models or tools. For API-based agents, however, they would require modifying the API-serving stack, which is typically infeasible for large services. TEEs also inherit strong trust assumptions in hardware vendors and enclave implementations, and remain vulnerable to key-extraction attacks~\cite{vanschaik2024SoKSGXFailHow}, necessitating careful evaluation of the threat model when considering their use.

\paragraph{TEE Proxies.}
A proxy mode, such as outlined in Town Crier~\cite{zhang2016TownCrierAuthenticated}, places a TEE between the client and a remote service, attesting to the authenticity of exchanged transcripts without requiring changes to the service itself. While this method offers low overhead and does not require any modifications to the service, it implies that all requests are routed in plaintext through the TEE rather than being sent directly from the client. This means that, unlike the earlier setting where we primarily used the TEE to prove the integrity of the executed computation, we now additionally require the TEE to maintain the confidentiality of the plaintext data it forwards, so as to prevent the proxy host from learning routed contents. This distinction matters because confidentiality can be considered a stronger assumption than integrity: compromising integrity can allow an attacker to exfiltrate data by altering the processing logic, while compromising confidentiality does not necessarily imply a compromise of integrity, as demonstrated by side-channel attacks~\cite{vanschaik2024SoKSGXFailHow,Kocher2018spectre,Lipp2018meltdown}. As a result, TEE Proxies are most attractive when the relayed data is public or low-sensitivity (e.g., price feeds) and routing through a proxy is acceptable.

\paragraph{Consensus and optimistic re-execution.}
Committee-based re-exe\-cu\-tion (e.g., blockchains or federated clusters) can ensure integrity under honest-majority assumptions~\cite{luu2015DemystifyingIncentivesConsensus}. However, in practice, repeated inference is prohibitively costly for large computations such as state-of-the-art LLMs, making these schemes suitable only for relatively small computations. Optimistic variants~\cite{conway2024OpMLOptimisticMachine} reduce repeated execution but introduce economic security assumptions and complex dispute-resolution procedures. Both approaches lack natural support for selective disclosure. They also require \emph{deterministic} execution, for example complicating interactions with stochastic and rapidly changing external data sources (e.g., price-feed APIs). And while we find them generally unsuitable for API-based agents, consensus-backed execution offers a compelling direction for fully on-chain agents, where each action is consensus-verified and autonomy derives directly from the underlying protocol.

\begin{table*}[t]
\centering
\footnotesize
\renewcommand{\arraystretch}{1.06}
\setlength{\tabcolsep}{6pt}
\caption{Comparison of candidate proof systems for VET components in the  API-based LLM agent setting.}
\label{tab:proof-comparison}
\resizebox{\textwidth}{!}{%
\begin{tabular}{@{} l p{0.2\textwidth} c p{0.26\textwidth} c @{}}
\toprule
\textbf{Proof system} &
\textbf{Applicability} &
\textbf{Deployability} &
\textbf{Trust assumptions} &
\textbf{Overhead} \\
\midrule

TEE (proxy) &
\raggedright API inference/tools (via proxy) &
{\centering \textbf{High} (no changes)\par} &
\raggedright Hardware vendor + enclave + target server &
{\centering \textbf{$\approx 1\times$}\par} \\

TEE (local) &
\raggedright Local inference/tools &
{\centering Medium (special hardware)\par} &
\raggedright Hardware vendor + enclave &
{\centering \textbf{$\approx 1$--$1.3\times$}\par} \\

Consensus re-execution &
\raggedright Small deterministic programs &
{\centering Low (determinism)\par} &
\raggedright Honest majority &
{\centering $N$ re-executions\par} \\

Optimistic re-execution &
\raggedright Deterministic programs &
{\centering Low (determinism)\par} &
\raggedright Honest majority + challenger &
{\centering $\approx 1\times$ (best case)\par} \\

Cryptographic proofs &
\raggedright Tiny local inference/tools &
{\centering Very low (rewrite into circuits)\par} &
\raggedright Cryptographic soundness &
{\centering $\approx10^2$--$10^4\times$\par} \\[4pt]

\midrule

\textbf{Web Proofs (notarized TLS transcripts)} &
\raggedright \textbf{API inference/tools} (no proxy) &
{\centering \textbf{High} (no changes)\par} &
\raggedright \textbf{Notary}$^\dagger$ + target server &
{\centering $\approx 2$--$5\times$\par} \\

\bottomrule
\end{tabular}}
\begin{flushleft}
\footnotesize{\textit{Note:} Overhead estimates are approximations and depend on specific workload. Web Proofs numbers reflect our empirical range (see Section~\ref{sec:evaluation}).}
\footnotesize{$^\dagger$ In TEE Proxies, forwarded data privacy and resulting transcript integrity are both conditional on the security of the underlying TEE. In Web Proofs, privacy with respect to the notary is provided cryptographically, while transcript integrity can be further hardened by running the notary inside a TEE (see Discussion ~\ref{par:notary-trust-model}).}
\end{flushleft}
\end{table*}

\section{Web Proofs as a Proof System}
\label{sec:webproofs}

\newcommand{\reqT}{\mathtt{req}_{\theta}}
\newcommand{\resT}{\mathtt{res}_{\theta}}

\newcommand{\Inject}[2]{\textsf{Inject}_{#1}(#2)}
\newcommand{\Parse}[2]{\textsf{Parse}_{#1}(#2)}
\newcommand{\HTTPReq}{\textsf{\small HTTP-Req}}
\newcommand{\HTTPResp}{\textsf{\small HTTP-Resp}}
\newcommand{\HTTP}[1]{\textsf{HTTP}(#1)}
\newcommand{\HTTPGet}[1]{\texttt{\small GET~#1}}
\newcommand{\JSON}[1]{\texttt{\small #1}}

The proof systems surveyed above target different trust models, and thus offer different trade-offs. In this work, we specifically focus on API-based LLM agents, where the remote model or tool API is assumed to behave correctly (that is, it is treated as a trusted service), while the untrusted party is the host running the agent orchestration code. Under this trust model, proof systems that attest the computation itself (for example, TEEs, ZKML proving, or consensus-based re-execution) remain applicable in principle to API-based components, since the orchestration host could forward verifiable proofs produced by the API service. In practice, however, this would require the API provider to generate and expose such proofs as part of the service, necessitating changes to the API-serving infrastructure, and would typically introduce overhead that scales with computation size. These approaches are thus often best suited for agent components that run locally (for example, local inference or tool execution), where the prover can be deployed alongside the computation.

The closest suitable solution for remote, black-box API components so far appear to be TEE Proxies. However, TEE Proxies still have several limitations: they are not transparent to the API infrastructure, since requests are routed through the proxy rather than originating from the host running the orchestration, and they require the proxy to handle plaintext requests and responses, making confidentiality dependent on the proxy's TEE security assumptions. TEE Proxies are therefore viable for components where forwarded content is public or low-sensitivity (e.g., public price feeds), but may be less suitable for components that handle secret-bearing data (e.g., proprietary prompts or API keys).

This motivates our use of TLS transcript proofs, referred to as \emph{Web Proofs}, as a practical choice that avoids these deployment and trust trade-offs with reasonable overheads. Web Proofs use MPC-assisted TLS~\cite{kalka2024ComprehensiveReviewTLSNotary,tlsnotaryteam2025TLSNotary} to generate cryptographic transcripts of standard HTTPS interactions (and thus require that the target API is TLS-enabled) by adjusting only the client-side logic. Their overhead scales with the size of the HTTPS transcript, and they preserve \emph{perfect secrecy with respect to the notary}, which learns nothing about the plaintext contents of the session. Web Proofs assume that the remote API behaves correctly and require no trust in the host that issues requests to the API, matching the trust assumptions we adopt for API-based LLM agents.

\subsection{Web Proof Sketch}\label{subsec:protocol-sketch}
A Client (the Prover) and a semi-trusted Notary jointly execute a TLS connection to a designated Target Server, which might be a remote API or a web service, such that (i) neither party can complete the session without the other, and (ii) only the Client learns the cleartext of the communication. Upon termination of the session, the Notary signs a succinct commitment to the transcript, yielding a proof $\pi_{\text{TLS}}$ attesting that the interaction indeed occurred with the specified Target Server.

\begin{figure}[htb]
    \centering
    \includegraphics[width=0.95\linewidth]{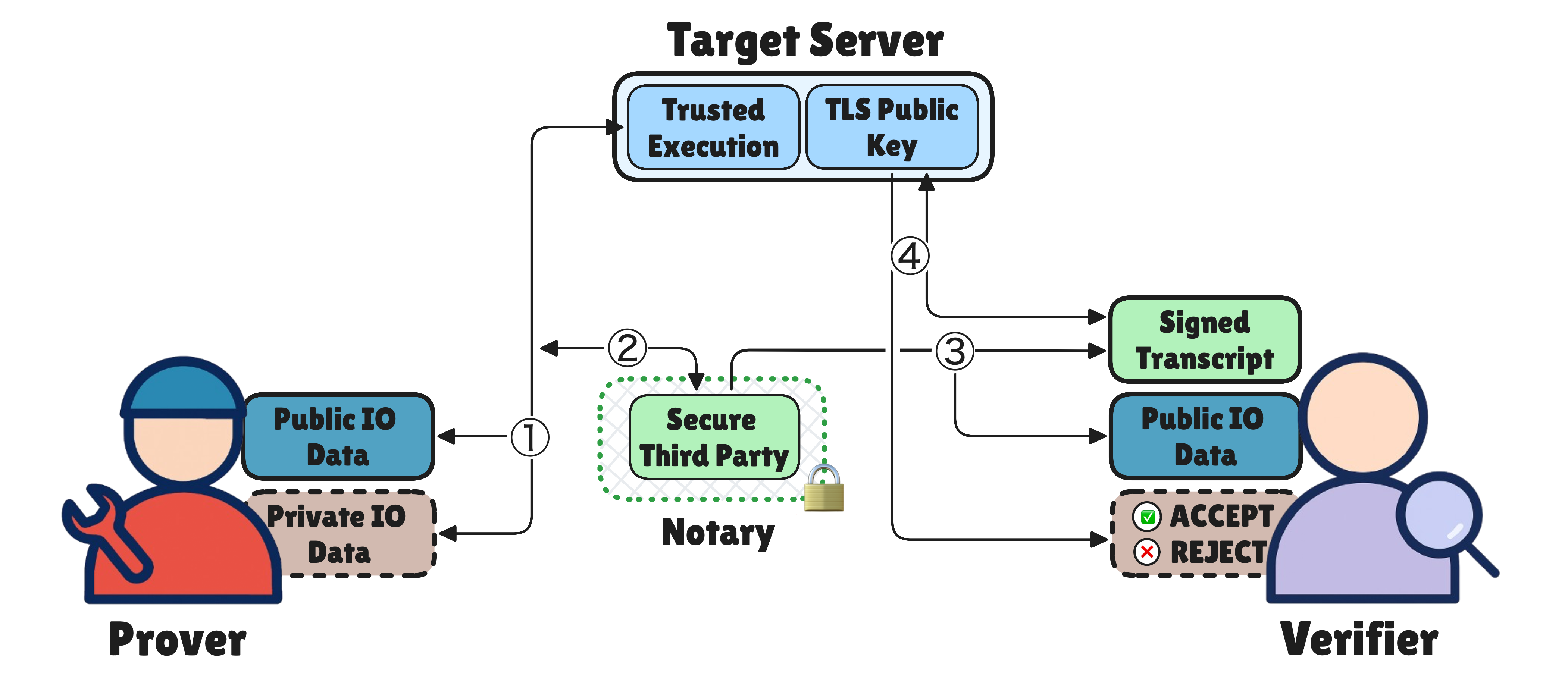}
    \caption{Web Proofs protocol overview.
    \textbf{(1) Connection:} The \emph{Prover} initiates a TLS session with the \emph{Target Server}.
    \textbf{(2) Co-execution:} An \emph{honest-but-curious Notary} jointly executes the MPC-TLS handshake, attesting to exchanged bytes without accessing plaintext.
    \textbf{(3) Transcript:} The Notary outputs a signed commitment, allowing the Prover to selectively disclose signed transcript portions to the \emph{Verifier}.
    \textbf{(4) Verification:} The Verifier checks the disclosed transcript against the Notary’s signature and the Server’s TLS public key, accepting if valid.}
    \Description{Diagram showing the Web Proofs protocol: a Prover connects to a Target Server, co-executes an MPC-TLS handshake with a Notary, obtains a signed transcript, and selectively discloses parts of it to a Verifier, who checks validity against the Notary’s signature and the Server’s TLS key.}
    \label{fig:webproof-diagram}
\end{figure}

\subsection{Security Properties}\label{subsec:security-properties}
Web Proofs provide the following guarantees, under the assumption that the Notary and the Target Server behave correctly, remain live, and keep their respective signing keys secret:

\begin{itemize}[leftmargin=1.2em]
    \item \textbf{Completeness.}
    The Prover can always produce a valid proof for any genuine TLS session with the Target Server.

    \item \textbf{Soundness.}
    No adversarial Prover can forge a proof of a conversation with the Target Server that will be accepted by the Verifier.

    \item \textbf{Privacy (selective disclosure).}
    The protocol ensures that the Notary does not learn the plaintext contents of the session or the identity of the Target Server, and that the Verifier sees only selectively disclosed portions of the transcript. This protects sensitive values (e.g., API keys, proprietary prompts) without weakening verifiability.
\end{itemize}

\paragraph{Notary trust model.}
\label{par:notary-trust-model}
While we need to assume an honest-but-curious Notary, this assumption is strictly weaker than the trust placed in hardware TEEs. This is because in practice, the Notary role can be realized in several ways: 
\begin{itemize}[leftmargin=1.5em]
    \item \emph{Public services} operated by independent organizations (e.g., the Ethereum Foundation’s TLSNotary service);
    \item \emph{Distributed Notaries} that split the role across multiple independent parties, similar to threshold randomness beacons such as the League of Entropy~\cite{taceo2025MultipartyNotariesZkTLS};
    \item \emph{TEE-assisted Notaries} that execute inside an enclave, hardened by hardware isolation  (e.g., Ethereum Foundation’s TLSNotary TEE service \footnote{\url{https://notary.pse.dev/v0.1.0-alpha.12-sgx/}}).
\end{itemize}
These instantiations reduce reliance on any single trusted machine or vendor, and importantly, do not require modifications to the target server. 

\subsection{Integrating Web Proofs into Agent Authentication}
\label{subsec:webproofs-integration}

Having established their security properties, we now show how Web Proofs instantiate the \emph{Component Proof Systems} required by our compositional authentication framework (Section~\ref{subsec:comp-auth}). Recall that an agent $\mathcal{A}=(\Core,\mathcal{T})$ consists of a reasoning core $\Core$ and a set of tools $\mathcal{T}$, each modeled by input–output relations over strings $\Sigma^\ast$ (Definitions~\ref{def:tool}--\ref{def:core}). In practice, both classes of components interact with external services through HTTPS APIs. Web Proofs allow these request–response interactions to be attested without modifying the remote service. To integrate them into our framework, we introduce encoding and parsing functions that map between the agent’s string-level view and the HTTP objects used by Web Proofs.

\paragraph{API tools.}
For a tool $t_d$ corresponding to a domain $d$, let
\[
\begin{aligned}
\Inject{t_d}{x} & : \Sigma^\ast \rightarrow \mathsf{HTTPReq},\\
\Parse{t_d}{\mathsf{HTTPResp}} & : \mathsf{HTTPResp} \rightarrow \Sigma^\ast
\end{aligned}
\]
denote deterministic encoding and parsing functions. A Web Proof transcript $(\reqT,\resT,\pi_{\text{WP}})$ certifies that $\resT$ is a genuine response from $d$ to request $\reqT$. By applying $\Inject{}{}$ and $\Parse{}{}$, we map these low-level objects back into the agent’s execution trace: given $\reqT=\Inject{t_d}{x}$ and $r=\Parse{t_d}{\resT}$, the transcript records that input $x$ to tool $t_d$ produced output $r$.

\paragraph{Cognitive cores.}
Similarly, for a API hosting LLM on domain $d$, we define
\[
\begin{aligned}
\Inject{\Core_h} &: \Sigma^\ast \rightarrow \mathsf{HTTPReq}, \\
\Parse{\Core} &: \mathsf{HTTPResp} \rightarrow
   \bigl(\Sigma^\ast, \mathcal{P}(\mathcal{V}_{\mathrm{tool}} \times \Sigma^\ast)\bigr).
\end{aligned}
\]
A Web Proof transcript $(\reqT,\resT,\pi_{\text{WP}})$, together with $\Inject{}{}$ and $\Parse{}{}$, certifies that $(y,T)=\Parse{\Core}{\resT}$ is a genuine output of the designated API endpoint in response to $\reqT$.

\paragraph{Verifier algorithm.}
Given candidate output $m$ and proof \\$(\reqT,\resT,\pi_{\text{WP}})$, the verifier proceeds as follows:
\begin{enumerate}[leftmargin=1.4em]
  \item Check $\pi_{\text{WP}}$ against the Notary’s public key and confirm that $(\reqT,\resT)$ is a genuine TLS transcript with domain $d$.
  \item Confirm that $\reqT,\resT$ are consistent with the deterministic $\Inject{}{}$ and $\Parse{}{}$ functions specified in the agent’s AID .
  \item Accept if $m$ matches the authenticated value (tool input or plaintext output) extracted from $\resT$.
\end{enumerate}

\paragraph{Discussion.}
This construction shows that Web Proofs can serve as drop-in Component Proof Systems: each tool or core invocation in an agent trace, performed over HTTPS, can be accompanied by a Web Proof, and the verifier only needs to check transcript validity and template consistency. In the Subsection~\ref{subsec:agent-identity}, we show how these verification parameters are explicitly represented in the agent’s Agent Identity Document (AID). % TODO - we can add a theorem if needed to show this formally

\subsection{Specifying Web Proofs in the AID}
\label{subsec:tls-api-vid-spec}

To make Web Proofs usable within our authentication framework, each TLS-based API component must be explicitly represented in the agent’s Agents Identity Document (AID) (cf.~ Figure~\ref{fig:sample-aid}). The AID entry specifies:
\begin{enumerate}[leftmargin=1.5em]
  \item the target server domain,
  \item the deterministic injection and parsing algorithms (identified by cryptographic hashes),
  \item the notary’s public verification key,
  \item and the protocol version of the Web Proof system(e.g.,\\ \texttt{TLSNotary-v0.1.0}).
\end{enumerate}

\subsection{Instantiation of Web Proofs}
\label{subsec:webproofs-instantiation}

To use Web Proofs within the VET framework, we instantiate them using TLSNotary~\cite{tlsnotaryteam2025TLSNotary}, which has both public notaries deployed (e.g., \url{https://notary.pse.dev/}) and allows to do self-hosted deployments.

\paragraph{Naïve baseline.}
A natural approach to integrating TLSNotary into VET is to open a single long-lived MPC-TLS channel per target domain, through which all subsequent API calls are routed. However, because TLSNotary setup costs scale linearly with channel capacity, provisioning a channel large enough to carry many future interactions leads to prohibitively high initial latencies. Moreover, for cognitive cores such as OpenAI’s tool-enabled APIs, each request transmits the entire dialogue history, causing significant retransmission overhead inside long-lived channels.

\begin{figure}[t]
  \centering
  \includegraphics[width=0.9\linewidth]{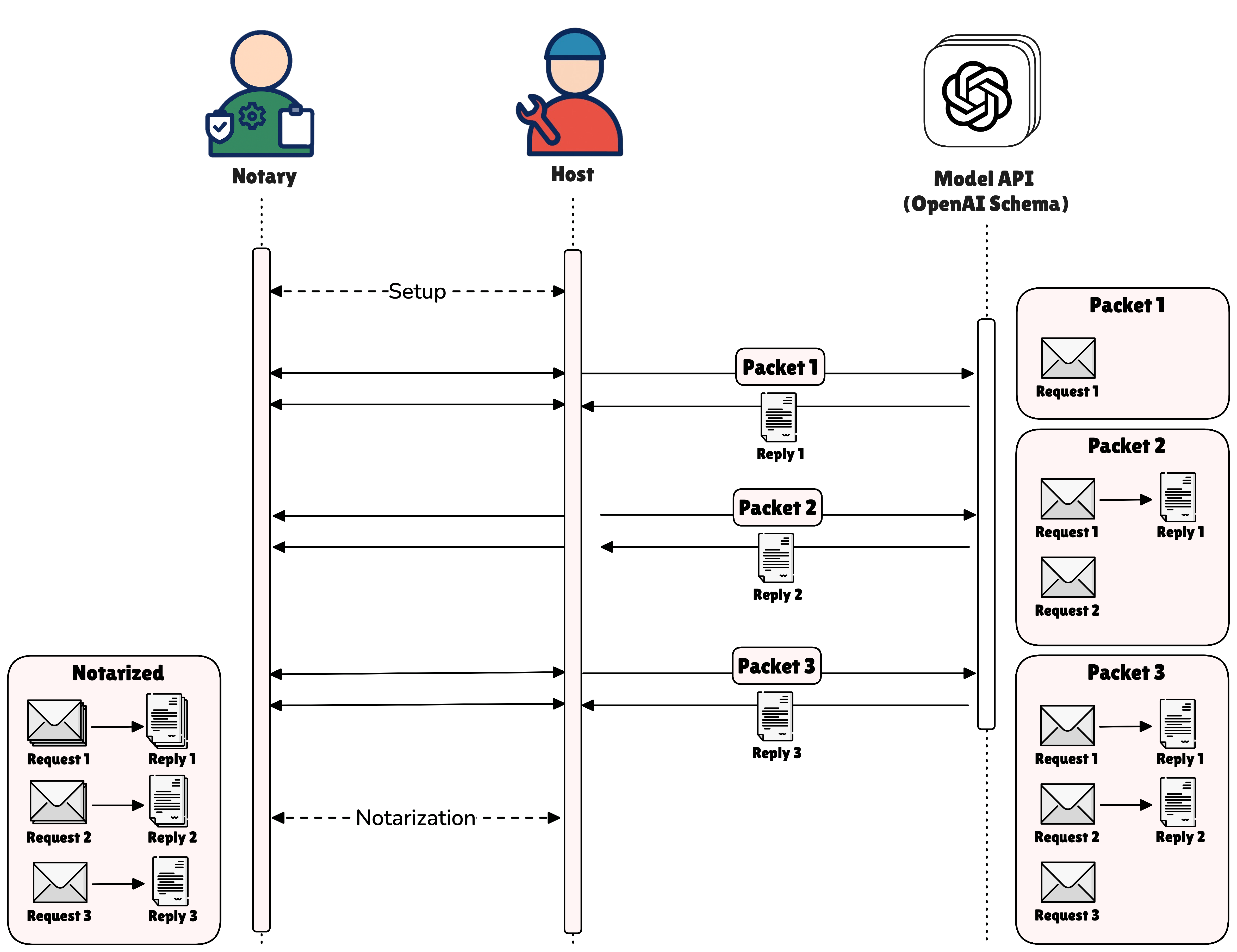}
  \caption{Instantiation of Web Proofs for API-based agent components.
Naïve long-lived channel incurs large one-time setup costs and retransmission overhead.}
  \Description{Diagram of the naïve Web Proofs setup: a long-lived communication channel is maintained between an agent and a verifier, leading to high initial setup costs and retransmission overhead during API-based interactions.}
  \label{fig:webproof-naive}
\end{figure}

\paragraph{Optimized channel strategy.}
Instead, we introduce another strategy, where for long-lived communication channels with model inference APIs we instead establish short-lived MPC-TLS channels on demand, with setup executed in parallel with ongoing established agent MPC-TLS channels. Specifically, while one request–response pair is being processed, the host pre-initializes additional channels with increasing capacity ($M, 2M, 3M, \dots$ bytes), so that the next request can immediately use a ready channel. This approach hides setup latency behind ongoing computation, avoids wasteful retransmission of repeated context, and reduces resource requirements on notaries to be able to maintain larger channels.

\begin{figure}[t]
  \centering
  \includegraphics[width=0.9\linewidth]{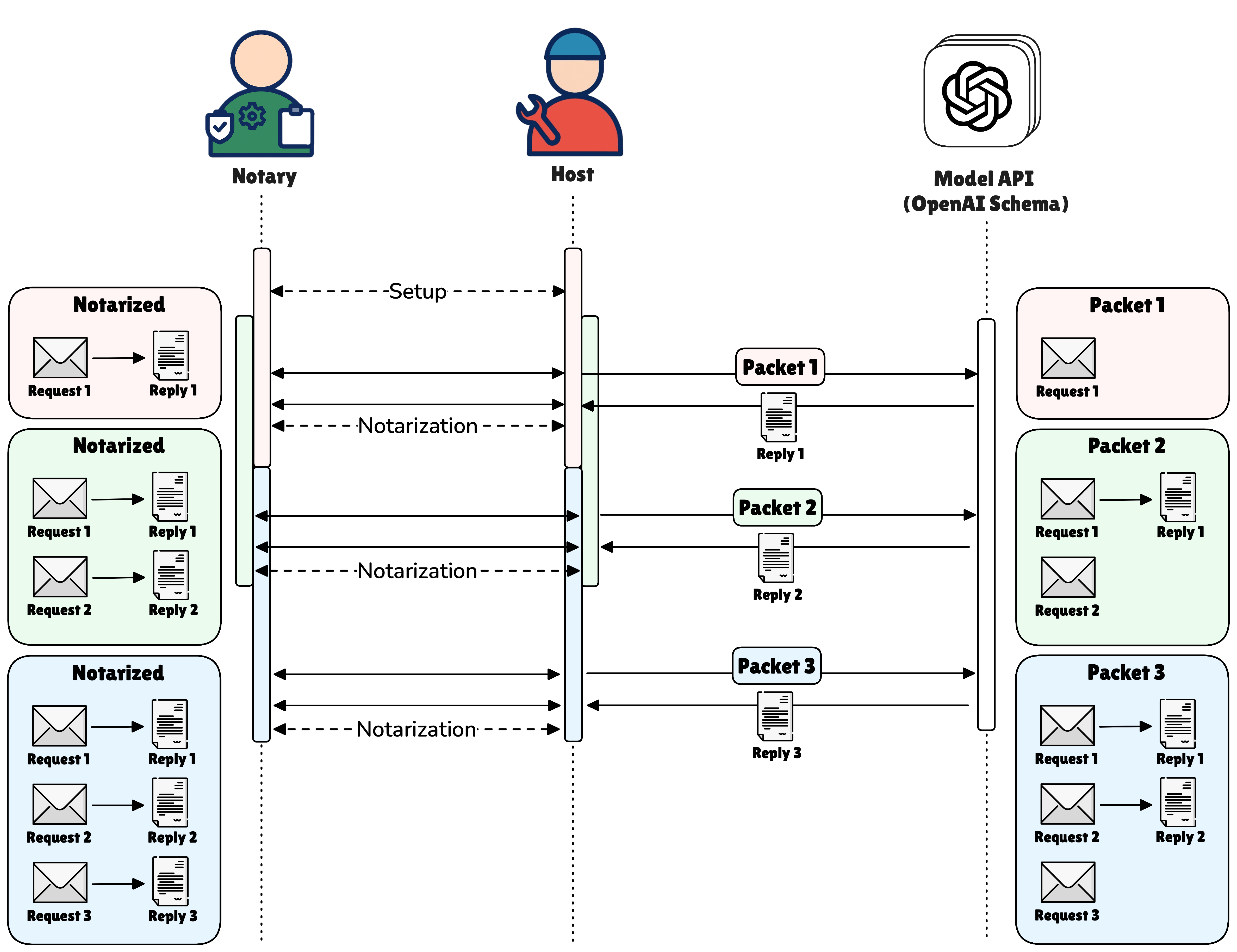}
  \caption{Instantiation of Web Proofs for API-based agent components.
Optimized strategy uses short-lived, parallel channels that amortize setup latency and minimize redundant bandwidth.}
  \Description{Diagram of the optimized Web Proofs setup: multiple short-lived parallel channels are created instead of a single long-lived one, reducing setup latency and avoiding repeated bandwidth overhead.}
  \label{fig:webproof-optimization}
\end{figure}

Importantly, this optimization preserves the security guarantees of Web Proofs and their integration into VET. Completeness and soundness hold as long as the notary executes the protocol correctly, while making the system practical at scale. As shown in Section~\ref{sec:evaluation}, the optimized strategy sustains multi-message sessions with overheads typically below $3\times$ compared to plain inference, whereas the naïve baseline fails to scale beyond a handful of messages.

% --------------------------------------------------------------------
\section{Evaluation}
\label{sec:evaluation}
% --------------------------------------------------------------------

To assess the practicality of Web Proofs within VET, we evaluate them as a component proof system for authenticating API-based LLM agents on realistic workloads. Our evaluation examines latency, scalability with session length, and deployment trade-offs across notary configurations and model-serving endpoints. These measurements allow us to determine whether Web Proofs meet the operational requirements of today’s API-based LLM agents and to validate that the optimized channel strategy is sufficient for multi-turn authenticated execution.

All experiments were performed on Google Cloud \texttt{c4.standard-4} instances (4 vCPU, 16~GB RAM, 5~Gbit/s network). We tested three notary configurations:  
(i) a public PSE Notary \footnote{\url{https://notary.pse.dev}} with 4~KB upload / 16~KB download per session;  
(ii) a self-hosted software-only (non-TEE) notary in Google Cloud \texttt{europe-west9-a} supporting 64~KB in both directions; and
(iii) an identically configured self-hosted TEE notary with the TDX extension enabled, allowing us to evaluate whether hardware isolation significantly affects notarization overhead. 

To contextualize overhead and provider variance, we benchmarked several real LLM endpoints hosted by Mistral (\texttt{ministral-3B}, \texttt{ministral-8B}), Anthropic (\texttt{Claude-Haiku-4.5}, \texttt{Claude-Sonnet\--4.5}), and RedPill (\texttt{openai/chatgpt-4o}). Unless otherwise specified, all results use \texttt{Claude-Haiku-4.5} with fixed-size prompts (500\,B request and 1\,KB response) to normalize inference cost across experiments. Benchmark code is available at our evaluation repository\footnote{\url{https://anonymous.4open.science/r/vet-your-agent-8E58/}}.
Table~\ref{tab:env} summarizes the environment.

\begin{table}[h]
\centering
\caption{Benchmark environment.}
\label{tab:env}
\scriptsize
\begin{tabular}{@{}lccc@{}}
\toprule
\textbf{Property} &
\textbf{Public PSE notary} &
\textbf{Software notary} &
\textbf{TEE notary} \\
\midrule
Region & Public (Europe) & europe-west9-a & europe-west9-a \\
RTT (host$\leftrightarrow$notary) & 9.9 ms & 5.0 ms & 4.2 ms \\
Bandwidth (up/down) & - & 5.36 / 4.96 Gbps & 4.62 / 3.56 Gbps \\
Max data / session (up/down) & 4 KB / 16 KB & 64 KB / 64 KB & 64 KB / 64 KB \\
\bottomrule
\end{tabular}
\end{table}

\paragraph{Protocol Comparison.}
\label{par:protocol-comparison}
We first compare a \emph{naïve} instantiation of Web Proofs, where state is carried across rounds in a single long-lived channel, with an \emph{optimized} design, where each message uses a fresh channel established in parallel. The optimized strategy reduces setup overhead from roughly $9.8$ s for a 6‑round naïve session to about $1.5$ s and does not require committing to the total number of messages in advance. More importantly, the naïve design quickly exhausts notary capacity: in our deployment it was effectively limited to six messages per session, whereas the optimized protocol sustained 32 rounds with the same notaries. The short-horizon per‑message latency of the optimized protocol is slightly higher (about $2.5$ s vs.\ $2.1$ s for the naïve variant), reflecting the cost of running channel setup in parallel with communication. However, this modest penalty is outweighed by the much lower setup time and the ability to support longer multi‑message sessions. We also expect that at longer horizons, retransmitting the full transcript in the naïve design would dominate the per‑message cost and eventually make naïve per‑message latency worse. Figure~\ref{fig:naive-vs-opt-haiku} summarizes these effects in terms of total session time as a function of dialogue length.

\paragraph{Scaling to Model Inference.}
We next explore how model size affects Web Proof latency. Using Mistral-hosted \texttt{ministral-3B} and \texttt{ministral-8B} APIs (with known parameter counts), we confirm that notarization overhead is driven by transcript size rather than model complexity. For single-shot, fixed-size prompts, total notarized latency is $2.69$ s for \texttt{ministral-3B} and $2.93$ s for \texttt{ministral-8B}, with overlapping 95\% confidence intervals despite the larger model.

\paragraph{Notary Overheads.}
\label{par:notary-overheads}
We compare total latency as a function of rounds for our Web Proof notaries against a direct API baseline and a TEE Proxy, all evaluated with the same fixed-size prompts (Figure~\ref{fig:proxy-notary-compare}, Table~\ref{tab:proxy-notary-overheads}). Each notary deployment in Table~\ref{tab:env} was also configured to host an HTTPS proxy so that proxy and notary share comparable network conditions. Across sampled providers and models, TEE Proxies add only about 1--20\% overhead relative to direct API calls, whereas TLS-notary runs introduce higher, round-dependent overheads. For the first message, when only a small initial channel is provisioned, notarized latency is typically 15--80\% higher than direct calls, depending on the underlying model latency. By Round~32, per-message latency under a TLS notary grows to roughly $\approx2\times$--$\approx7\times$ the direct-call latency, with the largest slowdowns for faster models where the recurring setup cost can no longer be hidden behind the API’s own computation time. As discussed in the protocol comparison above, this growth arises from transmitting increasingly large transcripts to the notary and incurring correspondingly higher setup costs. 

In practice, this growth in per-message overhead can be mitigated by periodic conversational summarization that compresses older turns, a technique already found in code assistants and agentic systems~\cite{IntroducingOperatorOpenAI,anthropic2025effective}. This helps cap context growth and keeps notarization overhead closer to the initial per-message cost.

\begin{figure}[t]
  \centering
  \begin{subfigure}[t]{0.485\linewidth}
    \centering
  \includegraphics[width=\linewidth]{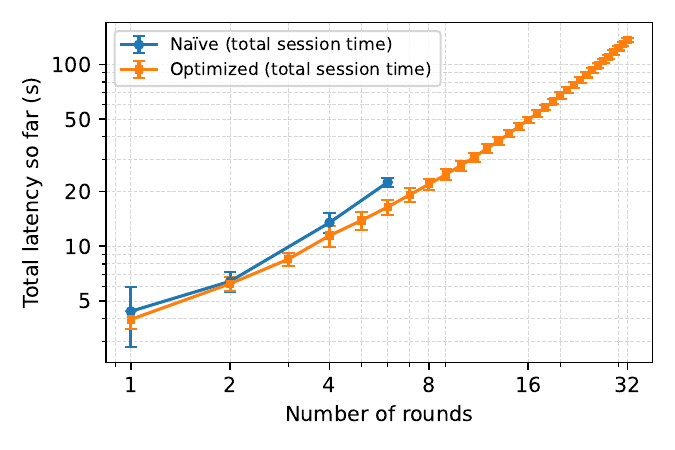}
      \caption{Total session time (including setup) for naïve long-lived sessions (1/2/4/6 rounds) and the optimized per-message TLS strategy (rolling cumulative time) on a log--log scale. Error bars show standard deviation over 10 runs.}
  \label{fig:naive-vs-opt-haiku}
  \end{subfigure}\hfill
  \begin{subfigure}[t]{0.485\linewidth}
    \centering
    \includegraphics[width=\linewidth]{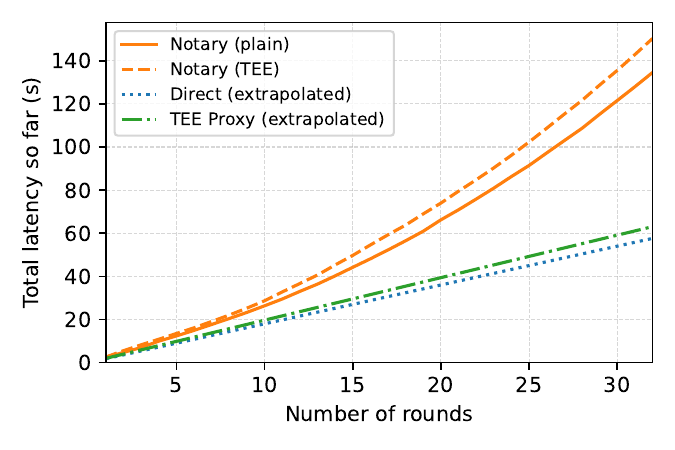}
    \caption{Total latency vs.\ rounds for direct API calls, a TEE-backed proxy, an optimized Web Proof notary, and a TEE-backed notary. Lines for Direct/TEE Proxy are extrapolated from a single round.}
    \label{fig:proxy-notary-compare}
  \end{subfigure}
  \caption{Web Proof application latency. (a) Optimized channels amortize handshakes and sustain multi-message sessions, while the naïve approach does not scale beyond a few messages. (b) Total latency vs.\ rounds for baseline (direct), TEE Proxy and Web Proof notaries (software and TEE-backed), highlighting the additional overhead of notarization relative to conventional proxies.}
  \Description{Two plots of Web Proof application latency. Subfigure (a) compares naïve long-lived vs.\ optimized short-lived sessions under different notaries, showing that only the optimized approach scales to longer multi-message sessions. Subfigure (b) shows total latency vs.\ rounds for direct calls, TEE Proxy, and Web Proof notaries (software and TEE-backed), illustrating the gap between proxy-based baselines and notarized execution.}
  \label{fig:combined-latency}
\end{figure}

\begin{table}[t]
\centering
\caption{Latency across models for a single 500\,B / 1\,KB message. Each entry shows mean latency $\pm$ std in milliseconds, with multiplicative overhead relative to the direct API call in parentheses. TLS-notary values additionally report the 32nd-round latency and overhead (first / 32nd).}
\label{tab:proxy-notary-overheads}
\scriptsize
\resizebox{\linewidth}{!}{%
\begin{tabular}{@{}lccc@{}}
\toprule
Model & Direct & TEE Proxy & TLS notary (1st / 32nd) \\
\midrule
OpenAI GPT-4o & $2180 \pm 96$ $(1.00\times)$ & $2290 \pm 89$ $(1.05\times)$ & $3300 \pm 437$ / $7300 \pm 559$ $(1.51\times / 3.35\times)$ \\
Claude-Haiku-4.5 & $1800 \pm 103$ $(1.00\times)$ & $1970 \pm 177$ $(1.09\times)$ & $2460 \pm 204$ / $6610 \pm 352$ $(1.37\times / 3.67\times)$ \\
Claude-Sonnet-4.5 & $4590 \pm 166$ $(1.00\times)$ & $4680 \pm 232$ $(1.02\times)$ & $5340 \pm 155$ / $8860 \pm 355$ $(1.16\times / 1.93\times)$ \\
Mistral 3B & $753 \pm 55$ $(1.00\times)$ & $756 \pm 52$ $(1.00\times)$ & $1330 \pm 110$ / $5940 \pm 557$ $(1.77\times / 7.89\times)$ \\
Mistral 8B & $901 \pm 46$ $(1.00\times)$ & $1060 \pm 123$ $(1.18\times)$ & $1590 \pm 135$ / $6030 \pm 325$ $(1.77\times / 6.70\times)$ \\
\bottomrule
\end{tabular}}
\end{table}

\paragraph{Comparison with Cryptographic Proofs.}
Table~\ref{tab:compare} situates Web Proofs alongside succinct zero-knowledge proofs of inference (SNARK-based ZKML). For example, for the Mistral-hosted \texttt{ministral-8B} API, a single notarized call under Web Proofs has total latency of about $2.93$\,s compared to $0.90$\,s for a direct call, a slowdown of roughly $3\times$ while keeping the model as a black-box service. By contrast, even for a much smaller LeNet-5 network with only 61k parameters (roughly $10^6\times$ fewer than modern LLMs), SNARK-based ZKML proofs take around 15\,s per \emph{token}. This several-orders-of-magnitude gap highlights that full cryptographic verification of inference remains substantially more expensive than notarizing TLS transcripts, and is better viewed as a complementary primitive that can be encoded in the AID for smaller computations when appropriate.

\begin{table}[ht]
\centering
\caption{Latency Comparison with Cryptographic Proofs.}
\label{tab:compare}
\scriptsize
\begin{tabular}{@{}lcc@{}}
\toprule
 & \textbf{Web Proofs (TLS notary)} & \textbf{ZKML (EZKL, LeNet-5 SNARK)} \\
\midrule
Model / size & Mistral \texttt{ministral-8B} API & LeNet-5 CNN (61k params) \\
Latency & $\approx$3s / message (naïve) & $\sim$15s / token \\
Scaling & transcript size & parameter count \\
\bottomrule
\end{tabular}
\end{table}

\paragraph{TEE-backed Notaries.}
\label{par:tee-backed-notaries}
Finally, we assess the effect of using TEE to reduce trust assumptions in the Notary. We compare our optimized per-message TLS notary running as a plain process with an otherwise identical deployment inside an Intel TDX enclave. Over 32-round sessions, total notarized latency increases from $136.1$\,s $\pm 9.1$\,s (software notary) to $152.0$\,s $\pm 5.2$\,s (TEE-backed), a modest $\approx\!1.12\times$ overhead. Setup time remains similar (1.50\,s vs.\ 1.70\,s), and per-round latency shifts from 4.20\,s to 4.70\,s on average. Unlike TEE Proxies, a TEE-backed notary preserves the perfect secrecy properties of Web Proofs, in the sense that the notary does not learn the plaintext contents of the TLS session, while remaining transparent to the model provider. These results indicate that TEE-backed notaries can strengthen integrity and confidentiality guarantees with minimal additional cost compared to a software-only notary.

\paragraph{Discussion.}
 Across both self-hosted and public endpoints, our evaluation shows that Web Proofs typically introduce less than a $3\times$ slowdown compared to direct API calls, with overhead driven primarily by transcript size and network placement rather than model size. In absolute terms, latencies remain in the range of a few seconds per message, which is acceptable for the low-frequency, high-value decisions where autonomous agents are most likely to be used (e.g., long-term trading, policy drafting, governance workflows), and where baseline LLM calls already take hundreds of milliseconds to seconds. In terms of trust, Web Proofs avoid introducing additional TEE hardware assumptions for plaintext confidentiality of communicated data, and can be strengthened by using distributed or TEE-backed notaries as discussed above. The primary performance bottlenecks are setup latency and limited channel capacity, which motivate both the optimized channel strategy and the transcript-compaction techniques described in this section.

% --------------------------------------------------------------------
\section{Case Study: VeriTrade}
\label{sec:case-study}
% --------------------------------------------------------------------

To illustrate the practical use of our framework, we implemented \textbf{VeriTrade}\footnote{VeriTrade code is available at \url{https://anonymous.4open.science/r/vet-your-agent-8E58/agent/}.}, an autonomous AI trading agent that produces, for each trade, a proof that its decision process was consistent with its declared configuration. VeriTrade mirrors existing ``autonomous'' portfolio-management agents, while showing how authentication can be attached to their outputs in a way that is independent of the orchestration host. The case study also explains proof-system choices for different VeriTrade components and shows how VET composes different proof-systems within a single Agent Identity Document (AID).

\paragraph{System Overview.}
VeriTrade fetches market data, generates trading decisions via an LLM core, and executes trades on a decentralized exchange (CoW Swap), submitting a verifiable execution trace alongside each trading decision. VeriTrade consists of:
\begin{itemize}[leftmargin=1.4em]
    \item \textbf{Market Data Tools (via TEE Proxy):} CoinGecko API for price and Polymarket API for sentiment data, proven via a TEE Proxy.
    \item \textbf{Cognitive Core (via Web Proofs):} \texttt{Claude-Haiku-4.5} via Anthropic's API, proven via Web Proofs.
\end{itemize}

\begin{figure}[t]
    \centering
    \includegraphics[width=0.9\linewidth]{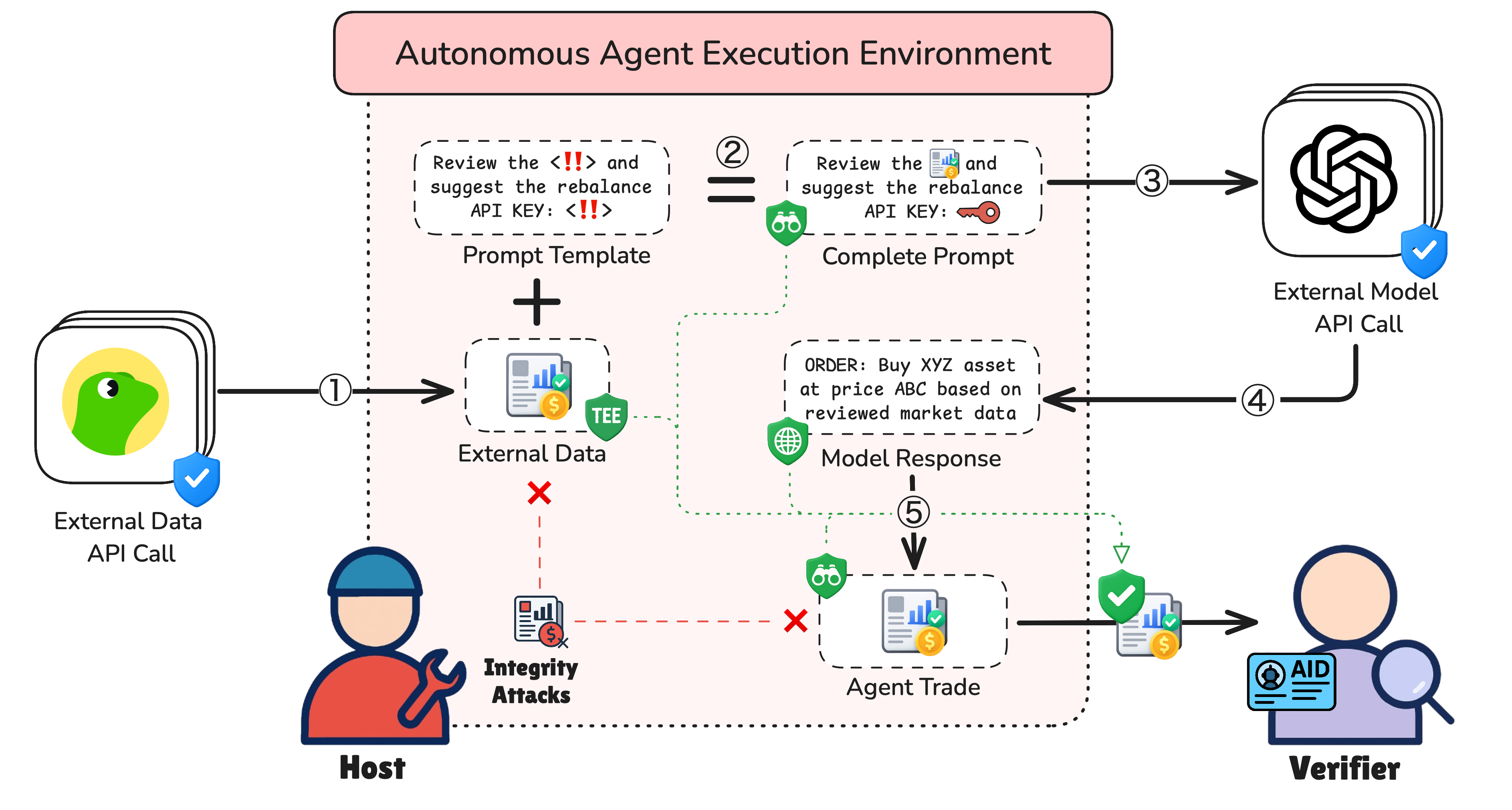}
    \caption{VeriTrade architecture. Market data is fetched via public APIs, proven using a TEE Proxy, while trade decisions are made by \texttt{Claude-Haiku-4.5}, proven using Web Proofs.}
    \Description{Diagram of the VeriTrade system architecture: public market APIs provide data verified by a TEE Proxy, a \texttt{Claude-Haiku-4.5} core generates trading decisions verified via Web Proofs, and the agent executes trades on a decentralized exchange.}
    \label{fig:veritrade_architecture}
\end{figure}

\paragraph{Selecting Proof Systems for VeriTrade Components.}
For the market-data API tools used by VeriTrade, we observe that requests and responses to these endpoints do not contain sensitive data, which means that a lower-overhead TEE Proxy is sufficient. In contrast, \texttt{Claude-Haiku-4.5} inference calls carry secret-bearing authentication and may include proprietary prompts, requiring careful consideration of the trade-offs. We benchmark both a TEE Proxy and Web Proofs in our deployment. Consistent with our evaluation in Paragraph~\ref{par:notary-overheads} (Figure~\ref{fig:proxy-notary-compare}, Table~\ref{tab:proxy-notary-overheads}), the TEE Proxy incurs lower overhead than Web Proofs: authenticated trading-decision latency is $4.20\,\mathrm{s}\,\pm\,0.71\,\mathrm{s}$, compared to $5.80\,\mathrm{s}\,\pm\,0.70\,\mathrm{s}$ (20 runs per setting). However, we find that the additional overhead is acceptable for periodic trading decisions and is not critical. Instead, the stronger secrecy of communicated data offered by Web Proofs motivates our choice to use Web Proofs for the cognitive core.

\begin{table}[h]
\centering
\caption{Choosing between TEE Proxies and Web Proofs for API-based agent components.}
\label{tab:choose-tee-vs-webproofs}
\scriptsize
\begin{tabular}{@{}lcc@{}}
\toprule
 & \textbf{TEE Proxy} & \textbf{Web Proofs} \\
\midrule
Transparency & No (sent through proxy) & Yes (no proxy) \\
Data secrecy & Weaker (TEE trusted with plaintext) & Stronger (notary learns no plaintext) \\
First-call overhead & $\approx$1--20\% (Table~\ref{tab:proxy-notary-overheads}) & $\approx$15--80\% (Table~\ref{tab:proxy-notary-overheads}) \\
\midrule
\textbf{Best used for:} & \textbf{Public / low-sensitivity data} & \textbf{Secret / proprietary data} \\
\bottomrule
\end{tabular}
\end{table}

\paragraph{Verifier.}
Proofs for each component are collected locally by the agent host and combined into a verifiable execution trace that is submitted alongside each trade decision. This enables auditability: anyone can authenticate (via a simple UI; see Appendix~\ref{app:veritrade-ui}) that trades were indeed produced by an agent conforming to the declared AID (see Appendix~\ref{app:veritrade-aid} for an excerpt), and not tampered with by the host.

\paragraph{Takeaway.}
VeriTrade's case study demonstrates two points. First, agent authentication via VET is already practical, with authenticated decision latencies in the order of seconds (e.g., $5.80\,\mathrm{s}\,\pm\,0.70\,\mathrm{s}$ for Web Proofs in our deployment). Second, we have practically demonstrated how VET can effectively combine multiple component proof systems, such as using Web Proofs for secret-bearing inference calls and a TEE Proxy for public market-data tools, selecting the most appropriate proof system for each agent component.

% --------------------------------------------------------------------
\section{Limitations and Future Work}
\label{sec:limitations_and_future_work}
% --------------------------------------------------------------------

Our work demonstrates the feasibility of host-independent \emph{authentication} for LLM-based agents, but important gaps remain before the stronger property of \emph{host-independent autonomy} can be realized. Authentication ensures that outputs match a declared configuration, yet a malicious host can still bias which outputs are revealed, delay disclosure, or steer behavior through covert timing and side channels. Closing this gap is a central objective for future work.

Several further limitations also point to concrete research directions. First, our current architecture assumes a honest-but-curious notary; compromise of its signing key would undermine both correctness, and reliance on a single entity creates centralization risks. Future designs should explore threshold-MPC or hybrid MPC-in-TEE deployments to distribute trust. Second, our scheme lacks freshness guarantees: replaying old proofs or selectively presenting stale outputs remains possible. Incorporating cryptographic timestamps or randomness beacons (e.g., Drand) would strengthen liveness and replay protection. Third, our formalism and implementation focus on a single reasoning core with sequential tool invocations, whereas emerging agents rely on parallel, clustered, or hierarchical structures. Extending compositional proofs to these richer execution models is an open challenge. Finally, practical deployment is constrained by the absence of standardized proof formats and tooling for verifiable bundles, which limits interoperability and adoption in sector-specific settings such as finance or governance.

In summary, VET establishes the first practical step toward host-independent autonomy, but future systems must reduce residual host influence, distribute trust in the notary, enforce freshness, and support more flexible agent architectures. Addressing these challenges will advance verifiable agents from authenticated outputs toward full autonomy in adversarial environments.
% --------------------------------------------------------------------
\section{Conclusion}
\label{sec:conclusion}
% --------------------------------------------------------------------

Autonomous agents are rapidly advancing from research prototypes to operating in both personal and public-facing capacities, yet the fact that their \emph{perception of autonomy} is in reality under \emph{complete host control} makes a dangerous situation, which has already caused trouble. This paper introduced VET (\emph{Verifiable Execution Traces}), the first formal framework for helping establish such \emph{host-independent autonomy}. VET shifts the trust anchor in agent outputs from the person who operates the host to the definition of the agent itself, providing a compositional way to bind each output to an Agent Identity Document via a Verifiable Execution Trace.

After surveying existing proof systems, we chose to instantiate VET with Web Proofs, a TLS-based proof system that works over black-box APIs without modifying model-serving infrastructure, and we complement it with a TEE Proxy for public, low-sensitivity tools where integrity is the primary requirement. Web Proofs satisfy soundness, completeness, and selective disclosure under realistic assumptions. Our evaluation demonstrates modest overheads, typically below $3\times$ compared to plain API-call time, even when applied to multi-billion-parameter models, and our \textsc{VeriTrade} case study illustrates feasibility in a high-stakes financial setting.

Our contributions mark the first practical step toward \emph{host-independent autonomy}. While we achieve robust authentication of agent outputs, realizing full autonomy requires further reducing assumptions about the host, mitigating covert timing and side-channel manipulation, and strengthening liveness and freshness guarantees. By formalizing agent authentication, demonstrating a practical instantiation, and explicitly framing the remaining gap to autonomy, this work lays a foundation for future research at the intersection of security, cryptography, and autonomous agent systems.

\begin{acks}
Simon Birnbach was supported by the Government Office for Science and the Royal Academy of Engineering under the UK Intelligence Community Postdoctoral Research Fellowships scheme.
Christian Schroeder de Witt was supported by a Royal Academy of Engineering Research Fellowship.
Artem Grigor was supported by an EZKL Fellowship (\url{https://www.ezkl.xyz}).
\end{acks}

\bibliography{zotero}

\appendix
\clearpage
\section{VeriTrade AID Excerpt}
\label{app:veritrade-aid}

\begin{figure}[h]
  \centering
  \begin{minipage}{0.98\linewidth}
  \lstset{style=asiaccs-json}
\begin{lstlisting}
{
  "agent_name": "VeriTrade",
  "core": {
    "model": "claude-haiku-4-5-20251001",
    "endpoint": "https://api.anthropic.com/v1/messages",
    "verification": {
      "TLSNotary": {
        "protocol_version": "v0.1.0-alpha.12",
        "notary_public_key": "ecdsa-secp256k1:047b48..."
      }
    }
  },
  "tools": [
    {
      "name": "FetchMarketSentiment",
      "endpoint": "https://gamma-api.polymarket.com/market",
      "verification": {
        "ProxyTEE": {
          "tee_type": "IntelTDX",
          "enclave_public_key": "ecdsa-secp256k1:048fa3..."
        }
      }
    },
    ...
  ]
}
\end{lstlisting}
  \end{minipage}
  \caption{Excerpt from VeriTrade’s AID (core inference via Web Proofs; market data tools via Proxy TEE).}
  \label{fig:trading-agent-card}
\end{figure}

\section{VeriTrade Verification UI}
\label{app:veritrade-ui}

\begin{figure}[H]
    \centering
    \includegraphics[width=0.9\linewidth]{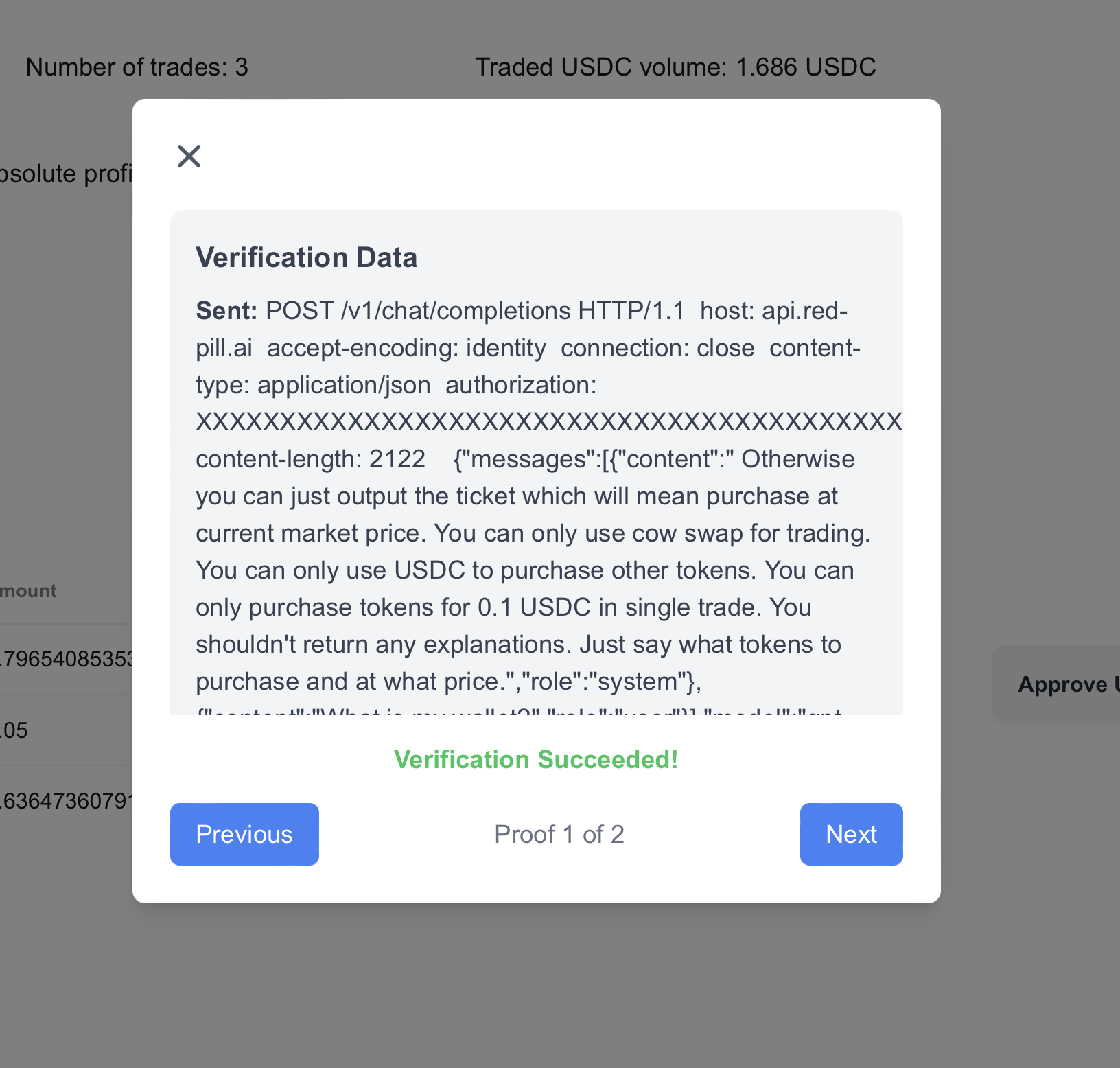}
    \caption{Prototype UI for inspecting the verified execution trace that accompanied VeriTrade's trading decisions.}
    \Description{Screenshot of a prototype VeriTrade interface that shows the notarized reasoning steps and actions of the agent, allowing the user to review them before executing a trade.}
    \label{fig:agent_verification_details}
\end{figure}

\end{document}